%% file: template.tex
\RequirePackage{fix-cm}
\documentclass[smallextended]{svjour3}       % onecolumn (second format)
\smartqed  % flush right qed marks, e.g. at end of proof
\usepackage{graphicx}
\usepackage{color}
\usepackage{amsmath}
\usepackage{amssymb}
\usepackage{tikz}
\usepackage{pgfplots}
\usepackage{multirow}
\usepackage{diagbox}
\newcommand{\Zq}{\ensuremath{\mathbb{Z}_q}}
\newcommand{\ve}[1]{\ensuremath{\boldsymbol{#1}}}

\newtheorem{construction}{Construction}

\begin{document}

\title{Duplication-Correcting Codes\thanks{This work was supported by the Institute for Advanced Study (IAS), Technische Universit\"{a}t M\"{u}nchen (TUM), with funds from the German Excellence Initiative and the European Union's Seventh Framework Program (FP7) under grant agreement no.~291763.\\
	Parts of this work have been presented at the 2017 Workshop on Coding and Cryptography (WCC), St. Petersburg \cite{Lenz17}.}
}


\author{Andreas Lenz \and Antonia Wachter-Zeh \and Eitan Yaakobi}


\institute{Andreas Lenz \at
              Institute for Communications Engineering, Technical University of Munich (TUM) \\
              \email{andreas.lenz@mytum.de}           %  \\
           \and
           Antonia Wachter-Zeh \at
           Institute for Communications Engineering, Technical University of Munich (TUM) \\
           \email{antonia.wachter-zeh@tum.de}
           \and
           Eitan Yaakobi \at
           Computer Science Department, Technion -- Israel Institute of Technology, Haifa, Israel\\
           \email{yaakobi@cs.technion.ac.il}
}

\date{}%Received: date / Accepted: date}

\maketitle

\begin{abstract}
In this work, we propose constructions that correct duplications of multiple consecutive symbols. These errors are known as tandem duplications, where a sequence of symbols is repeated; respectively as palindromic duplications, where a sequence is repeated in reversed order. We compare the redundancies of these constructions with code size upper bounds that are obtained from sphere packing arguments. Proving that an upper bound on the code cardinality for tandem \emph{deletions} is also an upper bound for \emph{inserting} tandem duplications, we derive the bounds based on this special tandem \emph{deletion} error as this results in tighter bounds. Our upper bounds on the cardinality directly imply lower bounds on the redundancy which we compare with the redundancy of the best known construction correcting arbitrary burst insertions.
Our results indicate that the correction of palindromic duplications requires more redundancy than the correction of tandem duplications and both significantly less than arbitrary burst insertions.
\end{abstract}

\section{Introduction}
The increasing demand for high density and long-term data storage and the recent advance in biotechnological methodology has motivated the storage of digital data in DNA. One interesting application in this area involves the storage of data in the DNA of living organisms. Tagging genetically modified organisms, infectious bacteria, conducting biogenetical studies or storing data are only a few in a list of modern applications. However, the data is corrupted by errors during the replication of DNA and therefore an adequate error protection mechanism has to be found. Typical errors include point insertions, deletions, substitutions and tandem or palindromic duplications. While the correction of substitutions, insertions and deletions is well studied, knowledge about correcting tandem and palindromic duplication errors is relatively limited. In the former case, a subsequence of the original word is duplicated and inserted directly after the original subsequence.
An example for a tandem duplication of length~$3$ in a DNA sequence $GATCATG$ is $GATC\underline{ATC}ATG$, where the underlined part highlights the duplication. Similarly, a palindromic duplication in the same word is $GATC\underline{CTA}ATG$. Note that these duplication errors are a special kind of burst insertion errors, where a consecutive sequence of random symbols is inserted into the transmitted sequence. In this paper, we propose several constructions that correct tandem, respectively palindromic duplication errors that yield lower redundancies than the best known burst insertion correcting codes. The redundancies of our constructions are further compared with lower bounds that we obtain from a sphere packing argument.
\subsection{Related Work}
Upper and lower bounds on the size of the largest code have been well studied for substitution errors in the Hamming metric. However, for other error models, such as insertions or deletions only recently non-asymptotic upper bounds onto code sizes have been found \cite{Fazeli15,Kulkarni13}. For repetition errors, which are special kinds of insertion errors, the only known upper bound is the asymptotic bound from Levenshtein \cite{Levenshtein65}. The non-asymptotic bounds in \cite{Kulkarni13} have been found by computing the fractional transversal number of the hypergraph associated with deletion errors. In \cite{Fazeli15}, this procedure has been analyzed and generalized to other error models, such as the $Z$-channel, grain-error channel, and projective spaces. Further, it has been shown that the average sphere packing value provides a valid upper bound on code sizes, if the associated hypergraph is regular and symmetric. Repetition errors form a related error model to tandem and palindromic duplications and corresponding error correcting codes have been studied in e.g. \cite{Dolecek10} and \cite{Levenshtein65}. More recently, an explicit construction for multiple repetition errors has been suggested \cite{Mahdavifar17}. These codes use the fact that repetition errors are equivalent to errors in the $\ell_1$-metric and are based on Lee-metric BCH codes \cite{Roth94}. Codes correcting tandem duplications have been considered in \cite{Jain16}, where amongst others a construction for the correction of an arbitrary number of fixed length duplications was presented. These codes are based on choosing irreducible words with respect to tandem duplications and their relation to zero run-length-limited systems has been illustrated. In this work we employ the method presented in \cite{Fazeli15}, known as the generalized sphere packing bound for tandem duplications to find non-asymptotic upper bounds on the code cardinalities and additionally present low-redundancy constructions that correct errors of these types, which extends the work in \cite{Lenz17}.
\subsection{Outline}
The paper is organized as follows. After introducing the preliminaries, in Section \ref{sec:relationship:duplication:deletion:codes}, we define duplication \emph{deletion} errors, an artificial error model that helps to find tight code size upper bounds. We derive upper bounds on the cardinality of codes that correct tandem duplication errors of length $\ell$ by finding a transversal on the associated hypergraph based on the error sphere size distribution in Section \ref{sec:upper:bounds:cardinalities}. These upper bounds directly imply a lower bound on the redundancy of a code. Finally, in the last section, we propose constructions that correct single tandem or palindromic duplications and compare their redundancies with the lower bounds.
\subsection{Preliminaries} \label{ss:preliminaries}
For two integers $a,b \in \mathbb{N}_0$, we write $\langle a \rangle_b$ to denote the integer rest of $a$ divided by $b$. We denote $\ve{x} = (x_1, x_2, \dots, x_n) \in \mathbb{Z}_q^n$ to be a vector of $n$ symbols over the ring of integers modulo $q$, $x_i \in \mathbb{Z}_q \, \forall \, i$. The length of a vector $\boldsymbol{x}$ is denoted by $|\ve{x}|$.

A \emph{tandem duplication} of length $\ell$ at position~$p$ with $0\leq p \leq n-\ell$ in a word $\ve{x} = (\ve{uvw})$, with $|\ve{u}| = p, |\ve{v}| = \ell, |\ve{w}| = n-\ell-p$ is defined by $\tau_{\ell}(\ve{x}, p) = (\ve{uvvw}) \in \mathbb{Z}_q^{n+\ell}$ and a \emph{palindromic duplication} of length $\ell$ is defined by $\rho_{\ell}(\ve{x},p) = (\ve{uv} \ve{v}^\mathrm{R} \ve{w})$, where $\ve{v}^\mathrm{R} = (v_\ell v_{\ell-1} \dots v_1)$ is the reversal of $\ve{v}$. The inverse operation, a \emph{tandem deletion} of length $\ell$ at position $0\leq p \leq n-2\ell$ in a word $\ve{x} = (\ve{uvvw})$ with $|\ve{u}| = p, |\ve{v}| = \ell, |\ve{w}| = n-2\ell-p$ is denoted by $\tau^D_{\ell}(\ve{x},p) = (\ve{uvw}) \in \mathbb{Z}_q^{n-\ell}$. Finally, we write a \emph{palindromic deletion} of length $\ell$ at position $0\leq p \leq n-2\ell$ in a word $\ve{x} = (\ve{uv}\ve{v}^\mathrm{R}\ve{w})$ with $|\ve{u}| = p, |\ve{v}| = \ell, |\ve{w}| = n-2\ell-p$ as $\rho^D_{\ell}(\ve{x},p) = (\ve{uvw}) \in \mathbb{Z}_q^{n-\ell}$.
\begin{example} [Tandem and palindromic duplication and deletion errors]
	Consider the word $\ve{x} = (11110220) \in \mathbb{Z}_3^{8}$. Then, a tandem duplication of length $2$ at position $p=3$ yields $\tau_2(\ve{x}, 3) = (11110\underline{10}220)$, where the underlined part is the erroneous duplication. Similarly, a palindromic duplication of length $2$ at position $3$ results in $\rho_2(\ve{x}, 3) = (11110\underline{01}220)$. Examples for tandem and palindromic deletion errors of length $2$ in $\ve{x}$ are $\tau^D_2(\ve{x}, 0) = (110220)$ and $\rho^D_2(\ve{x}, 4) = (111102)$.
\end{example}
Note that the deletion operations are only defined at positions $p$, where the word $\ve{x}$ is of the form $(\ve{uvvw})$, respectively $(\ve{uv}\ve{v}^\mathrm{R}\ve{w})$ with $|\ve{u}| = p$.

With these definitions, the \emph{sphere} of a word $\ve{x}$ with radius $t$ is the set of all vectors that are reached by \emph{exactly} $t$ tandem or palindromic duplications, respectively deletions, i.e.,
\begin{equation} \label{eq:error:sphere}
S^\epsilon_t(\ve{x}) = \left\{\ve{y} | \ve{y} = \epsilon\left(...(\epsilon(\ve{x}, p_1)...), p_t\right)\right\},
\end{equation}
where $\epsilon$ is the error type and $p_i$ denote the position of the $i$-th duplication, respectively deletion. Here,
\begin{itemize}
	\item $\epsilon = \tau_{\ell}$ for tandem duplication errors,
	\item $\epsilon = \tau^D_{\ell}$ for tandem deletion errors,
	\item $\epsilon = \rho_{\ell}$ for palindromic duplication errors,
	\item or $\epsilon = \rho^D_{\ell}$ for palindromic deletion errors.
\end{itemize}
We further define the error \emph{ball}
\begin{equation} \label{eq:error:ball}
B^\epsilon_t(\ve{x}) = \left\{\ve{y} | \ve{y} = \epsilon\left(...(\epsilon(\ve{x}, p_1)...), p_\theta\right), \theta \leq t\right\},
\end{equation}
as the set of all vectors that can be reached by \emph{at most} $t$ errors. Note that interestingly the size of these sets depends on $\ve{x}$, which is the key complication when computing upper bounds on the code cardinality.

For a word~$\ve{x}$, let $r(\ve{x})$ be the number of \emph{runs}, $r_i(\ve{x})$ the length of the $i$-th run, respectively $r^{(i)}(\ve{x})$ the number of runs of length~$i$ and $r^{(\geq i)}(\ve{x})$ be the number of runs of length at least $i$ in $\ve{x}$.
\begin{example}[Runs]
	Consider the word $\ve{x} = (11110220)$, which has $r(\ve{x}) = 4$ runs. The lengths of those runs are $r_1(\ve{x}) = 4$, $r_2(\ve{x}) = 1$, $r_3(\ve{x}) = 2$ and $r_4(\ve{x}) = 1$. Therefore, there are $r^{(1)}(\ve{x}) = 2$ runs of length $1$, $r^{(2)}(\ve{x}) = 1$ run of length $2$, $r^{(3)}(\ve{x}) = 0$ runs of length $3$ and $r^{(4)}(\ve{x}) = 1$ run of length $4$.
\end{example}
The $\ell_1$ norm of a vector $\ve{x} \in \mathbb{N}_0^n$ over the natural numbers is given by the sum of its entries and is denoted by $|\ve{x}|_1 = \sum_{i=1}^{n}x_i$.
\begin{definition} \label{eq:t:error:correcting}
	A codebook $\mathcal{C} \subset \mathbb{Z}_q^n$ is called a $t$-\emph{tandem duplication} (\emph{palindromic duplication}, \emph{tandem deletion}, \emph{palindromic deletion}) correcting codebook, if $B^\epsilon_{t}(\ve{c}_1) \cap B^\epsilon_{t}(\ve{c}_2) \neq \emptyset$ implies $\ve{c}_1 = \ve{c}_2$ for all $\ve{c}_1, \ve{c}_2 \in \mathcal{C}$.
\end{definition}
In the following we will use the term \emph{single}-error correcting for the case $t=1$.
\subsubsection{Tandem Duplication Preliminaries} \label{ss:tandem:duplication:preliminaries}
The following definitions are based on the findings in \cite{Jain16} and show the connection between tandem duplications and $\ell_1$-metric errors, which is helpful for both, deriving upper bounds onto code cardinalities and finding code constructions that correct tandem duplications.
\begin{definition}[$\ell$-step derivative]\label{def:l:step:derivative}
	For $\ve{x} \in \mathbb{Z}_q^n$ we define the $\ell$-step derivative $\phi_\ell(\ve{x})=(\ve{u}_x, \ve{v}_x)$ with $\ve{u}_x = (x_1, x_2, \dots, x_\ell)$ and $\ve{v}_x = (x_{\ell+1}, x_{\ell+2}, \dots, x_n) - $ $(x_1, x_2, \dots, x_{n-\ell})$.
\end{definition}
Note that in the following we refer to $\ve{v}_x$ as the second part of the $\ell-$step derivative of $\ve{x}$, as introduced in Definition \ref{def:l:step:derivative}. It has been shown in \cite{Jain16} that a tandem duplication of length $\ell$ in $\ve{x}$ corresponds to an insertion of $\ell$ consecutive zeros in $\ve{v}_x$. This motivates the introduction of the $\ell$-\emph{trunk} and $\ell$-\emph{zero-signature} representation for $\ve{v}_x$. 
\begin{definition}[$\ell$-Trunk, $\ell$-zero signature]
	Denote by $0^m$ the $m$-fold repetition of $0$ and let $\ve{v} = (0^{m_0},w_1,0^{m_1},w_2, \dots, w_{p}, 0^{m_p})$ with $w_i \in \mathbb{Z}_q \setminus \{0\}$ and $p = wt_{\mathrm{H}}(\ve{v})$ be the Hamming weight of $\ve{v}$. We define the $\ell$-\emph{trunk} of $\ve{v}$ to be $\mu_\ell(\ve{v}) = (0^{\langle m_0 \rangle_\ell},w_1,0^{\langle m_1 \rangle_\ell},w_2, \dots, w_{p}, 0^{\langle m_p \rangle_\ell})$ as the word that is obtained by shortening every zeros run of length $m$ to be of length $m \bmod \ell$. Further, the $\ell$-\emph{zero signature} of $\ve{v}$ is defined as $\sigma_\ell(\ve{v}) = \left(\left\lfloor\frac{m_0}{\ell}\right\rfloor, \left\lfloor\frac{m_1}{\ell}\right\rfloor, \dots, \left\lfloor\frac{m_p}{\ell}\right\rfloor\right)$.
\end{definition}
By this definition, the $\ell$-zero signature is a vector over the natural numbers $\mathbb{N}_0$ and counts the number of distinct $\ell$ consecutive $0$'s in one run of consecutive $0$'s in $\ve{v}$. Note that $\ve{v}$ is uniquely determined by its $\ell$-trunk $\mu_\ell(\ve{v})$ and $\ell$-zero-signature $\sigma_\ell(\ve{v})$. It is easy to see that a tandem duplication in $\ve{x}$ corresponds to increasing an entry of $\sigma_\ell(\ve{v}_x)$ by 1 (a tandem deletion corresponds to decreasing the entry by 1), but leaves the root $(\ve{u}_x, \mu_\ell(\ve{v}_x))$ unchanged.
\begin{example}
	Let $\ve{x} = (21010121) \in \mathbb{Z}_3^8$ be a word of length $8$. Its $\ell$-step derivative $\phi_{}(\ve{x}) = (\ve{u}_x, \ve{v}_x)$ for $\ell=2$ is given by $\phi_2(\ve{x}) = ((21), (100020))$. Hence, the $\ell$-trunk is $\mu_2(\ve{v}_x) = (1020)$. The $\ell$-zero signature has length $wt_{\mathrm{H}}(\ve{v})+1 = 3$ and is given by $\sigma_2(\ve{v}_x) = (010)$. The word $\ve{x}$ is now perturbed by a tandem duplication of length $2$, resulting in $\ve{y} = \tau_2(\ve{y}, 0) = (21\underline{21}010121)$. Computing the $\ell$-step derivative yields $\phi_2(\ve{y}) = (\ve{u}_y, \ve{v}_y) = ((21), (\underline{00}100020))$. The $\ell$-trunk computes to $\mu_2(\ve{v}_y) = (1020)$ and the $\ell$-zero signature of $\ve{v}_y$ is $\sigma_2(\ve{v}_y) = (110)$. As expected, the first entry of the $\ell$-zero signature increased by $1$.
\end{example}
The notation is summarized in the Table \ref{tab:notation:tandem:duplication}.
\begin{table}[h]
	\centering
	\begin{tabular}{c|c}
		Notation & Definition \\
		\hline
		$\tau_\ell(\ve{x}, p)$ & Tandem duplication of length $\ell$ at position $p$ \\
		$\tau^D_\ell(\ve{x}, p)$ & Tandem deletion of length $\ell$ at position $p$ \\
		$\rho_\ell(\ve{x}, p)$ & Palindromic duplication of length $\ell$ at position $p$ \\
		$\rho_\ell^D(\ve{x}, p)$ & Palindromic deletion of length $\ell$ at position $p$ \\
		$S_t^\epsilon(\ve{x})$ & Error sphere of $t$ errors of type $\epsilon$ \\
		$B_t^\epsilon(\ve{x})$ & Error ball of $t$ errors of type $\epsilon$ \\
		$\phi_\ell(\ve{x})$ & $\ell$-step derivative ($\phi_\ell(\ve{x}) = (\ve{u}_x, \ve{v}_x)$) \\
		$\mu_\ell(\ve{v}_x)$ & $\ell$-trunk \\
		$\sigma_\ell(\ve{v}_x)$ & $\ell$-zero signature

	\end{tabular}
	\caption{Summary of notation}
	\label{tab:notation:tandem:duplication}
\end{table}
\section{Relationship between Duplication and Deletion Codes}
\label{sec:relationship:duplication:deletion:codes}
We start with revealing relationships between tandem duplication correcting codes with tandem deletion correcting codes. Although the tandem deletion error is an artificial model, it will help later to formulate tight bounds on codes correcting tandem duplication errors.
\subsection{Equivalence of Tandem Duplication and Deletion Codes} \label{ss:equivalence:tandem:duplication:deletion:codes}

For conventional insertion and deletion correcting codes, it is known that a code $\mathcal{C}$ is $t$-insertion correcting if and only if it is $t$-deletion correcting \cite{Levenshtein66}. A similar behavior can be shown for tandem duplications, which is formulated in the following theorem.
\begin{theorem} \label{thm:equivalence:tandem:duplication:deletion}
	A code $\mathcal{C} \subset \mathbb{Z}_q^n$ is $t$-tandem duplication (length $\ell$) correcting if and only if it is $t$-tandem deletion (length $\ell$) correcting.
\end{theorem}
\begin{proof}
	By Definition \ref{eq:t:error:correcting} it is sufficient to show that the tandem duplication error balls for all $\ve{x}, \ve{y} \in \Zq^n$ intersect if and only if their tandem deletion error balls intersect, i.e.
	\[ B^{\tau_\ell}(\ve{x}) \cap B^{\tau_\ell}(\ve{y}) \neq \emptyset \Longleftrightarrow B^{\tau^D_\ell}(\ve{x}) \cap B^{\tau^D_\ell}(\ve{y}) \neq \emptyset \, \forall \, \ve{x}, \ve{y} \in \Zq^n, \]
	in order to prove Theorem \ref{thm:equivalence:tandem:duplication:deletion}. As illustrated in the previous section, a tandem duplication does not change the $\ell$-trunk of a word and increases one entry of the $\ell$-zero signature by $1$. Similarly, a tandem deletion does not change the $\ell$-trunk of a word and decreases one entry of the $\ell$-zero signature by $1$. Therefore, writing $|\bullet|_1$ as the $\ell_1$-norm, $\phi_{\ell}(\ve{x}) = (\ve{u}_x, \ve{v}_x)$ and $\phi_{\ell}(\ve{y}) = (\ve{u}_y, \ve{v}_y)$, it follows.
	\begin{align*}
		&B_t^{\tau_\ell}(\ve{x}) \cap B_t^{\tau_\ell}(\ve{y}) \neq \emptyset \\
		\Longleftrightarrow\,&(\ve{u}_x, \mu_\ell(\ve{v}_x)) = (\ve{u}_y, \mu_\ell(\ve{v}_y)) \land |\sigma_\ell(\ve{v}_x)-\sigma_\ell(\ve{v}_y|_1 \leq 2t \\
		\Longleftrightarrow\,& B_t^{\tau^D_\ell}(\ve{x}) \cap B_t^{\tau^D_\ell}(\ve{y}) \neq \emptyset.
	\end{align*}
	 \qed
\end{proof}
\subsection{Relationship between Palindromic Duplication and Deletion codes}
For palindromic duplication errors, an equivalence similar to Theorem \ref{thm:equivalence:tandem:duplication:deletion} does not hold. A counter example for $\ell=2, t=1$ that shows that not every palindromic \emph{deletion} correcting code is palindromic \emph{duplication} correcting is presented here.
\begin{example}
	Let $\mathcal{C} = \{\ve{c}_1, \ve{c}_2\}$ with $\ve{c}_1 = (010101)$ and $\ve{c}_2 = (010011)$. $\mathcal{C}$ is single palindromic \emph{deletion} correcting, since $B^{\rho_2^D}_1(\ve{c}_1) = \{\ve{c}_1\}$ and $B^{\rho_2^D}_1(\ve{c}_2) = \{\ve{c}_2, (0101)\}$ and thus $B^{\rho_2^D}_1(\ve{c}_1) \cap B^{\rho_2^D}_1(\ve{c}_2)= \emptyset$. On the other hand, $\mathcal{C}$ is not single palindromic \emph{duplication} correcting since $B^{\rho_2}_1(\ve{c}_1) \cap B^{\rho_2}_1(\ve{c}_2) = \{(01001101)\}$.
\end{example}
The following example illustrates that also not every palindromic \emph{duplication} correcting code is palindromic \emph{deletion} correcting.
\begin{example}
	Consider the code $\mathcal{C} = \{\ve{c}_1, \ve{c}_2\}$ with $\ve{c}_1 = (011010)$ and $\ve{c}_2 = (011110)$. $\mathcal{C}$ is single palindromic \emph{duplication} correcting, since 
	\begin{align*}
		B^{\rho_2}_1(\ve{c}_1) &= \{\ve{c}_1, (01101010), (01111010), (01100110), (01101100), (01101001)\}, \\ B^{\rho_2}_1(\ve{c}_2) &= \{\ve{c}_2, (01101110), (01111110), (01111001)\},
	\end{align*}
	and thus $B^{\rho_2}_1(\ve{c}_1) \cap B^{\rho_2}_1(\ve{c}_2)= \emptyset$. However, $\mathcal{C}$ is not single palindromic \emph{deletion} correcting since $B^{\rho_2^D}_1(\ve{c}_1) \cap B^{\rho_2^D}_1(\ve{c}_2) = \{(0110)\}$.
\end{example}
\section{Upper Bounds on the Code Cardinalities} \label{sec:upper:bounds:cardinalities}
One of the most basic problems in coding theory is finding the largest code correcting a given type of error $\epsilon$. In general, this problem can be stated as
\[
A^\epsilon(n,t) = \underset{\mathcal{C} \subseteq \Zq^n}{\max} \, |\mathcal{C}|, \quad \mathrm{s.t.} \quad B^\epsilon_t(\ve{c}_1) \cap B^\epsilon_t(\ve{c}_2) = \emptyset, \, \forall \, \ve{c}_1, \ve{c}_2 \in \mathcal{C},
\]
where $A^\epsilon(n,t)$ denotes the maximum cardinality of a code of length $n$ that corrects $t$ errors of type $\epsilon$. Due to the fact that the exact number $A^\epsilon(n,t)$ is often not known, one is interested in finding tight upper and lower bounds onto this number. In the following, we derive non-asymptotic upper bounds for tandem and palindromic duplication errors by using tools from hypergraph theory similar to the approach from \cite{Fazeli15,Kulkarni13}.
\subsection{Upper Bound for General Error Types}
Consider the hypergraph $\mathcal{H}_{n,t}^\epsilon = \left(\mathcal{V}_{n,t}^\epsilon, \mathcal{E}_{n,t}^\epsilon\right)$ with vertices $\mathcal{V}_{n,t}^\epsilon$ and hyperedges $\mathcal{E}_{n,t}^\epsilon$,
\begin{align*}
	\mathcal{V}_{n,t}^\epsilon &= \bigcup_{\ve{x} \in \Zq^n} B_t^\epsilon(\ve{x}), \\
	\mathcal{E}_{n,t}^\epsilon &= \{ B_t^\epsilon(\ve{x}), \ve{x} \in \Zq^n \},
\end{align*}
that is associated with a channel of at worst $t$ errors of type $\epsilon$ in words of length $n$. The vertices $\mathcal{V}_{n,t}^\epsilon$ of the hypergraph $\mathcal{H}_{n,t}^\epsilon = \left(\mathcal{V}_{n,t}^\epsilon, \mathcal{E}_{n,t}^\epsilon\right)$ consist of all possible channel inputs and outputs, while the hyperedges $\mathcal{E}_{n,t}^\epsilon$ represent possible channel outcomes for a selected channel input $\ve{x} \in \Zq^n$. The following definitions of hypergraph transverals and matchings are naturally associated with problems in coding theory and can be found in, e.g. \cite{Hansen81}.
\begin{definition}[Hypergraph matching]
	A matching of a hypergraph $\mathcal{H} = (\mathcal{V}, \mathcal{E})$ is a set of disjoint hyperedges $\mathcal{M} \subseteq \mathcal{E}$, such that $E_i \cap E_j = \emptyset$ for all $E_i, E_j \in \mathcal{M}$.
\end{definition}
With this definition, a matching is described by a function $M: \mathcal{E} \rightarrow \{0,1\}$, that satisfies
\[\sum_{E \in \mathcal{E}: \ve{\nu} \in E } M(E) \leq 1, \quad \forall \, \ve{\nu} \in \mathcal{V}, \]
where $M(E)=1$ indicates that a hyperedge $E$ is included in the matching and $M(E)=0$ means that the hyperedge is not included in that matching $M$.
\begin{definition}[Hypergraph transversal]
	A transversal of a hypergraph $\mathcal{H} = (\mathcal{V}, \mathcal{E})$ is a set of vertices $\mathcal{T} \subseteq \mathcal{V}$, such that for each hyperedge $E \in \mathcal{E}$, there exists a $\ve{\nu} \in E$ with $\ve{\nu} \in \mathcal{T}$.
\end{definition}
Similarly to the hypergraph matching, a transversal can therefore be described by a function $T: \mathcal{V} \rightarrow \{0,1\}$, with
\begin{equation} \label{eq:transversal:condition}
\sum_{\ve{\nu} \in E} T(\ve{\nu}) \geq 1, \quad \forall \, E \in \mathcal{E},
\end{equation} %
where $T(\ve{\nu})$ indicates, whether a vertex $\ve{\nu}$ is included in the transversal. With these definitions, finding the code with maximum cardinality is equivalent to finding the largest matching $\mathcal{M}^\star$. The solution to this problem is referred to as the \emph{matching number} and is denoted by $\nu(\mathcal{H}) = |\mathcal{M}^\star|$. Consequently, $\nu(\mathcal{H}_{n,t}^\epsilon) = A^\epsilon(n,t)$. Simplifying the computationally intensive problem of finding the exact matching number, it has been shown in \cite{Kulkarni13} that the matching number is upper bounded by any \emph{fractional} transversal, which is a function $T^*: \mathcal{V} \rightarrow \mathbb{R}_0^+ $, satisfying the transversal condition \eqref{eq:transversal:condition}. Hence, we restate the following Lemma from \cite{Fazeli15,Kulkarni13} which gives an upper bound on the maximum code cardinality $A^\epsilon(n,t)$.
\begin{lemma}\label{lemma:cardinality:upper:bound}
	Let $\mathcal{H}^\epsilon_{n,t} = (\mathcal{V}_{n,t}^\epsilon, \mathcal{E}_{n,t}^\epsilon)$ be a hypergraph that is associated with $t$ errors of type $\epsilon$. The maximum code cardinality $A^\epsilon(n,t)$ for a code of length $n$ correcting $t$ errors of type $\epsilon$ is upper bounded by
	\[
	A^\epsilon(n,t) \leq \sum_{\ve{\nu} \in \mathcal{V}} T^\epsilon(\ve{\nu}),
	\]
	where $T^\epsilon: \mathcal{V}_{n,t}^\epsilon \rightarrow \mathbb{R}_0^+$ is a fractional transversal which satisfies
	\begin{equation} \label{eq:fractional:transversal:condition}
		\begin{aligned}
		\sum_{\ve{\nu} \in B^\epsilon_t(\ve{x})} &T^\epsilon(\ve{\nu}) \geq 1, \quad \forall \, \ve{x} \in \Zq^n,\\
		&T^\epsilon(\ve{\nu}) \geq 0, \quad \forall \, \ve{\nu} \in \mathcal{V}.
		\end{aligned}
	\end{equation}
\end{lemma}
With Lemma \ref{lemma:cardinality:upper:bound} it is possible to formulate upper bounds on codes correcting tandem duplication errors by using an appropriate fractional transversal. However, notice that the transversal sum in Lemma \ref{lemma:cardinality:upper:bound} is formulated over all words $\ve{\nu} \in \mathcal{V}$, where $\mathcal{V}$ contains words of length $n, n+\ell, \dots, n+t\ell$ for the duplication errors. If, in contrary, bounds for deletion errors are derived, the vertices of the hypergraph are words of length $n, n-\ell, \dots, n-t\ell$. This indicates that a bound based on the deletion errors is smaller than for the corresponding duplication error. Indeed, this is observed also for classical deletions and insertions where this method provides a better bound for deletions than for insertions. As shown in Section \ref{sec:relationship:duplication:deletion:codes} it holds that $A^{\tau_\ell}(n,t) = A^{\tau^D_\ell}(n,t)$ and. Therefore, a fractional transversal for the hypergraphs associated with tandem \emph{deletion} errors provides valid upper bounds onto the size of tandem \emph{duplication} error correcting codes.

In the following we show how to formulate fractional transversals that yield upper bounds for tandem deletion duplication error correcting codes. The next definition will be helpful for the upcoming steps.
\begin{definition}[Irreducible words]
	For an error type $\epsilon \in \{\tau_\ell, \rho_\ell\}$, we define the set of all $t$-irreducible words to be
	\[ \mathrm{IRR}^\epsilon_t = \{ \ve{\nu} \in \Zq^*: S^{\epsilon^D}_t(\ve{\nu}) = \emptyset \}. \]
\end{definition}
Note that in contrast to substitution errors and conventional deletion errors, it is possible that the error spheres for duplication deletion errors are empty. Therefore, the fractional transversal that will serve for upper bounding the sizes of our codes has to be formulated carefully and will be only non-zero for vectors that are either irreducible or contained in the error sphere of a word $\ve{x} \in \Zq^n$.
\begin{lemma} \label{lemma:fractional:transversal:tandem}
	For some fixed $t$ and $\ell$, the function $T^{\tau_\ell^D}: \mathcal{V}^{\tau_\ell^D}_{n,t} \rightarrow \mathbb{R}_0^+$
	\[ T^{\tau_\ell^D}(\ve{\nu}) = \left\{ \begin{array}{cl}
	0, & \mathrm{if} \, \ve{\nu} \notin \Zq^{n-t\ell} \land \ve{\nu} \notin \mathrm{IRR}^{\tau_\ell}_t \\
	1, & \mathrm{if} \, \ve{\nu} \in \mathrm{IRR}^{\tau_\ell}_t \\
	\big|S^{\tau^D_\ell}_t(\ve{\nu})\big|^{-1}, & \mathrm{if} \, \ve{\nu} \in \Zq^{n-t\ell} \land \ve{\nu} \notin \mathrm{IRR}^{\tau_\ell}_t \\
	\end{array} \right.. \]
	is a fractional transversal for the hypergraph $\mathcal{H}^{\tau_\ell^D}_{n,t} = (\mathcal{V}_{n,t}^{\tau_\ell^D}, \mathcal{E}_{n,t}^{\tau_\ell^D})$ associated with $t$ tandem deletion errors of length $\ell$.
\end{lemma}
\begin{proof}
	To show that $T^{\tau_\ell^D}$ is a valid transversal, we need to proof that the fractional transversal condition \eqref{eq:fractional:transversal:condition} is satisfied for all $\ve{x} \in \Zq^n$. Consider first the case that $B^{\tau^D_\ell}_t(\ve{x}) \cap \mathrm{IRR}^{\tau_\ell}_t \neq \emptyset$, which means that the error ball around $\ve{x}$ contains an irreducible word. Then the transversal condition is directly fulfilled, as $T^{\tau_\ell^D}(\ve{\nu}) = 1$ for at least one element $\ve{\nu} \in B^{\tau^D_\ell}_t(\ve{x}) $. For the case $B^{\tau^D_\ell}_t(\ve{x}) \cap \mathrm{IRR}^{\tau^D}_t = \emptyset$, we will first show that
	\begin{equation} \label{eq:tandem:monotonicity}
		\big|S^{\tau^D_\ell}_t(\ve{\nu})\big| \leq \big|S^{\tau^D_\ell}_t(\ve{x})\big|
	\end{equation}
	for all $\ve{\nu} \in B^{\tau^D_\ell}_t(\ve{x})$. This inequality is known as the monotonicity property \cite{Fazeli15} and can be proven using the expression for the error sphere size, which will be derived in Lemma \ref{lemma:tandem:deletion:sphere}. Since the length of the $\ell$-zero signature of $\ve{x}$ and $\ve{\nu}$ is the same, i.e. $|\sigma_\ell(\ve{v}_x)| = |\sigma_\ell(\ve{v}_\nu)|$, and $\sigma_\ell(\ve{v}_\nu) \leq \sigma_\ell(\ve{v}_x)$, inequality in \eqref{eq:tandem:monotonicity} follows.
	Hence, the transversal sum satisfies
	\[ \sum_{\ve{\nu} \in S^{\tau^D}_t(\ve{x})} \big|S^{\tau^D_\ell}_t(\ve{\nu})\big|^{-1} \geq |S^{\tau^D_\ell}_t(\ve{x})| \underset{\ve{\nu} \in S^{\tau_\ell^D}_t(\ve{x})}{\min} \big|S^{\tau^D_\ell}_t(\ve{\nu})\big|^{-1} \geq 1. \]
	Notice that $\big|S^{\tau^D_\ell}_t(\ve{\nu})\big|^{-1}$ is well defined, as $B^{\tau^D_\ell}_t(\ve{x})$ contains no irreducible words in this case.
	\qed
\end{proof}
A key ingredient for the proof of Lemma \ref{lemma:fractional:transversal:tandem} is the monotonicity property \eqref{eq:tandem:monotonicity} of tandem deletion errors, which means that the deletion sphere sizes for all words in a deletion sphere are smaller than the size of the parent sphere. Computing the overall transversal sum with the functions as given in Lemma \ref{lemma:fractional:transversal:tandem}, we obtain the following upper bound on the maximum cardinalities of tandem duplication correcting codes.
\begin{corollary} \label{cor:cardinality:upper:bound}
	Denote by $N^{\tau_\ell^D}_t(n,i) = |\{ \ve{x} \in \Zq^{n}: |S^{\tau_\ell^D}_t(\ve{x})| = i \}|$. Then the maximum cardinality of any $t$-tandem duplication correcting code is upper bounded by
	\[ A^{\tau_\ell}(n,t) \leq \sum_{i=0}^{t} |\mathrm{IRR}^{\tau_\ell}_t \cap \Zq^{n-i\ell}| + \sum_{i=1}^{i_{\max}} \frac{N^{\tau^D_\ell}_t(n-t\ell,i)}{i},  \]
	where $i_{\max}$ is the maximum error sphere size for $t$ tandem deletion errors.
\end{corollary}
Note that it has been shown in \cite{Jain16} that the number of irreducible words $\mathrm{IRR}^{\tau_\ell}_t$ is connected to the number of run-length-limited words for tandem duplication errors. Therefore, for large $n$, the second summand in Corollary \ref{cor:cardinality:upper:bound} dominates the upper bound.
\subsection{Bound for Tandem Deletions}
To find explicit expressions for the bound stated in Corollary \ref{cor:cardinality:upper:bound}, we have to compute the number of words of length $n$ with sphere size $i$ for tandem and palindromic deletion errors. In the following, we will find combinatorial expressions for these numbers that can then directly be used to obtain code size upper bounds. Expressions for the sphere sizes of the discussed error types can be found in Appendix \ref{sec:error:spheres:tandem:palindromic:dupliations:deletions}.
\begin{lemma} \label{lemma:tandem:deletion:number:words:sphere:size} The number of words of length $n$ with tandem deletion sphere size $i$ is given by
	\begin{align*}
	N^{\tau_\ell^D}_1(n,i) &= |\{ \ve{x} \in \Zq^{n}: |S^{\tau_\ell^D}_1(\ve{x})| = i \}| = \\
	&= \sum_{\nu=i}^{\left\lfloor\frac{n}{\ell}\right\rfloor-1}  \sum_{\omega=i-1}^{n-(\nu+1) \ell} q^\ell A(n - (\nu+1) \ell,\ell-1,\omega) \binom{\omega+1}{i} \binom{\nu-1}{i-1},
	\end{align*}
	where $A(n',\ell',\omega)$ is the number of all words $\ve{x} \in \mathbb{Z}_q^{n'}$ that have zero-runs of length at most $\ell'$ and Hamming weight $\omega$. 
\end{lemma}
\begin{proof}
	We consider the $\ell-$step derivative $\phi_\ell(\ve{x}) = (\ve{u}, \ve{v})$. According to Corollary~\ref{cor:single:tandem:deletion:sphere}, the size of the single tandem deletion sphere is given by $|S^{\tau^D_\ell}_1(\ve{x})| = wt_\mathrm{H}(\sigma_\ell(\ve{v}))$ and we therefore want to find the number of words $\ve{x} \in \mathbb{Z}_q^n$ with $wt_\mathrm{H}(\sigma_\ell(\ve{v})) = i$.
	
	Let $\nu$ be the number of length $\ell$ tandem duplications in $\ve{x}$, i.e. $|\sigma_\ell(\ve{v})|_1 = \nu$. Further let $\mathcal{J}$ denote the support set of $\sigma_\ell(\ve{v})$, i.e. $\mathcal{J} = \{m : \sigma(\ve{v})_m \neq 0\}$, with $|\mathcal{J}| = i$. The number of possibilities to distribute the duplications into $\sigma_\ell(\ve{v})$ for a given support $\mathcal{J}$ is equal to the number of solutions of
	\begin{equation}
	\sum_{j=1}^{i}y_j = \nu, \quad y_j \in \mathbb{N}, \,\,\forall \,\, 1\leq j \leq i.
	\end{equation}
	This number is given by $\binom{\nu-1}{i-1}$ \cite[Lemma 2.2]{Kulkarni13}. Further, let $\omega$ be the Hamming weight of the $\ell$-trunk, i.e. $wt_{\mathrm{H}}(\mu_\ell(\ve{v})) = \omega$ and thus $|\sigma_\ell(\ve{v})| = \omega+1$, which corresponds to the number of unambiguous positions for tandem duplications of length $\ell$. The number of possible support sets $\mathcal{J}$ of $\sigma_\ell(\ve{v})$ with $|\mathcal{J}| = i$ then is $\binom{\omega+1}{i}$. The vector $\mu_\ell(\ve{v})$ can be chosen to be any $q$-ary vector of length $n-(\nu+1)\ell$ that has zero-runs of length at most $\ell-1$ and Hamming weight $\omega$. The number of such vectors is given by $A(n - (\nu+1) \ell,\ell-1,\omega)$. Finally, the first $\ell$ symbols $\ve{u}\in \mathbb{Z}_q^\ell$ can be chosen arbitrarily and thus have $q^\ell$ possibilities. \qed
\end{proof}
It can be deduced from the results in \cite{Kurmaev11} that for $\omega\geq 2$ the number of all $q$-ary vectors of length $n'$, maximum zero-run length $\ell'$ and weight $\omega$ is given by
\begin{equation*}
A(n',\ell',\omega) = (q-1)^\omega \left\{ \begin{array}{ll} \multirow{2}{*}{$n'>\ell':$} &  \sum\limits_{p=0}^{\ell'}\sum\limits_{j=0}^{\omega-1}  (-1)^j \big( 
\binom{\omega-1}{j}
\binom{n'-p-1-j(\ell'+1)}{\omega-1}-  \\ & \quad \quad \quad \quad \quad \quad \binom{n'-p-1-(j+1)(\ell'+1)}{\omega-1}\big),  \\
n'\leq \ell': &  \binom{n'}{\omega},
\end{array} \right. .
\end{equation*} 
For $\omega=0$ and $\omega=1$, it holds that $A(n',\ell',0) = 1$, if $n'\leq \ell'$ and $A(n',\ell',0) = 0$ otherwise. Further, $A(n',\ell',1) = (q-1)\max \{0, 2(\ell '+1)-n'\}$. \\

Plugging in the result from Lemma \ref{lemma:tandem:deletion:number:words:sphere:size} into Corollary \ref{cor:cardinality:upper:bound} directly gives an upper bound on the cardinality of a code correcting a single tandem duplication error of size $\ell$. Note that the number of irreducible words of length $n$ can be obtained by $|\mathrm{IRR}^{\tau_\ell}_t \cap \Zq^{n}| = N^{\tau_\ell^D}_1(n,0)$.
\section{Code Constructions}
In this section, we propose code constructions that can correct a single tandem duplication, respectively a single palindromic duplication.
\subsection{Code Correcting a Single Tandem Duplication} \label{ss:code:single:tandem:duplication}
For the following construction, which is able to correct a single tandem duplication of length $\ell$, we use the general construction presented in \cite{Jain16} with an explicit code which can correct a single error in the $\ell_1$ metric. Varshamov-Tenegolts (VT) codes \cite{Varshamov65} are single asymmetric error correcting codes that can also be applied to single $\ell_1$-metric errors. According to the original definition, we construct a set of vectors over the natural numbers, which satisfies the VT constraint.
\begin{definition} For some $0 \leq a \leq n$, the set of vectors satisfying the VT constraint is defined as
	\[ \mathcal{VT}_a(n) = \left\{ \ve{x} \in \mathbb{N}_0^n : \left\langle\sum_{i=1}^n x_i\right\rangle_{n+1} = a  \right\}.\]
\end{definition}
\begin{construction} \label{con:tandem:vt} For some $\ve{a} \in \mathbb{N}^{n-\ell+1}, 1 \leq a_i \leq i$, 
	\[\mathcal{C}_1(n) = \{ \ve{x} \in \Zq^n : \sigma_\ell(\ve{v}) \in \mathcal{VT}_{a_{\omega+1}}(\omega+1) \}, \]
	where $\phi_\ell(\ve{x})=(\ve{u}, \ve{v})$ is the $\ell$-step derivative of $\ve{x}$ and $\omega = wt_\mathrm{H}(\ve{v})$ is the length of the $\ell$-zero signature.
\end{construction}
It can be directly deduced from the results in \cite{Jain16} that Construction \ref{con:tandem:vt} is single tandem duplication correcting. The minimum size of this construction can directly be obtained and is given in the following Lemma.
\begin{lemma}
		There exist integers $\ve{a} \in \mathbb{N}^{n-\ell+1}, 1 \leq a_i \leq i$, such that the cardinality of Construction \ref{con:tandem:vt} satisfies
		\[ |\mathcal{C}_1(n)| \geq q^\ell \sum_{\nu=0}^{\left\lfloor\frac{n}{\ell}\right\rfloor-1}  \sum_{\omega=0}^{n-(\nu+1) \ell} A(n - (\nu+1) \ell,\ell-1,\omega) \frac{\binom{\omega+\nu}{\nu}}{\omega+2}. \]
\end{lemma}
\begin{proof}
	Let $\nu$ be the number of length $\ell$ tandem duplications in $\ve{x}$, i.e. $|\sigma_\ell(\ve{v})|_1 = \nu$. Further, let $\omega$ be the Hamming weight of the $\ell$-trunk (for its definition, see Section \ref{ss:tandem:duplication:preliminaries}), i.e. $wt_{\mathrm{H}}(\mu_\ell(\ve{v})) = \omega$ and thus the length of the $\ell$-zero signature is given by $|\sigma_\ell(\ve{v})| = \omega+1$. The total number of such possible $\ell$-zero signatures is $\binom{\omega+\nu}{\nu}$. Due to the pigeonhole principle, we can always find an integer $a_{\omega+1} \leq \omega+1$, such that the number of $\ell$-zero signatures satisfying the VT-constraint is at least $\frac{\binom{\omega+\nu}{\nu}}{\omega+2}$. Counting the $q$-ary vectors of length $n-(\nu+1)\ell$ that have zero-runs of length at most $\ell-1$ and Hamming weight $\omega$ with $A(n - (\nu+1) \ell,\ell-1,\omega)$ yields the lemma. \qed
\end{proof}
\subsection{Construction Correcting a Palindromic Duplication for $\ell=2$} \label{ss:code:single:palindromic:duplication}
For the case of a binary alphabet $q=2$, we propose a construction that is able to correct a single palindromic duplication of length $\ell=2$. We start with some definitions that will help for defining the code construction.
\begin{definition}[$1$-run-length profile]
	The set of all binary words of length $n$, whose number of runs of length $1$ is congruent to $a \bmod 5$, is defined by
	\begin{equation*}
	\mathcal{D}_a(n) = \left\{ \ve{x} \in \mathbb{Z}_2^n | \left\langle r^{(1)}(\ve{x}) \right\rangle_5 = a \right\},
	\end{equation*}
	where $r^{(1)}(\ve{x})$ is the number of runs with length $1$.
\end{definition}
\begin{definition}[Run-length profile constraint]
	We define the set of binary words with run-length profile constraint $b$
	\begin{equation*}
	\mathcal{E}_b(n) = \left\{ \ve{x} \in \mathbb{Z}_2^n | \left\langle C(\ve{x})\right\rangle_{2n+1} = b \right\},
	\end{equation*}
	with the checksum for the run lengths
	\begin{equation*}
	C(\ve{x}) = \sum_{i=1}^{r(\ve{x})} i r_i(\ve{x}).
	\end{equation*}
\end{definition}
In general, it is possible to formulate the above definitions for words over arbitrary finite alphabets, however, in this section we are only interested in finding a construction for binary words. With these definitions, it is possible to state the following construction, which we will show that is a single palindromic duplication (length $2$) correction correcting code.
\begin{construction} \label{con:palindrome:2} For $a \in \{0,1,2,3,4\}$ and $b \in \{0,1,\dots, 2n\}$, we construct the following binary code of length $n$
	\[ \mathcal{C}_2(n) = \mathcal{D}_a(n) \cap \mathcal{E}_b(n). \]
\end{construction}
\begin{theorem}
	The code $\mathcal{C}_2(n)$ is single palindromic duplication (length $2$) correcting for any $a \in \{0,1,2,3,4\}$ and $b \in \{0,1,\dots, 2n\}$.
\end{theorem}
\begin{proof}
	Let us consider all possible constellations for palindromic duplications of length $2$. There are five basic patterns that have to be taken into account and they are displayed in Table \ref{tab:Duplication_Constellations} with their corresponding erroneous outcomes.
	\begin{table}[h]
		\centering
		\begin{tabular}{l|cl|cl}
			Case & Original Sequence & RL Profile & Perturbed Sequence & RL Profile \\ \hline
			1.a & $\underline{aa}$ & $(r_j)$ & $\underline{aaaa}$ & $(r_j+2)$ \\
			1.b & $\underline{ab}$ & $(r_j,1)$ & $\underline{abba}$ & $(r_j,2,1)$ \\ \hline
			2.a & $\underline{ab}aa$ & $(r_j,1,r_{j+2})$ & $\underline{abba}aa$ & $(r_j,2,r_{j+2}+1)$ \\
			2.b & $\underline{ab}a$ & $(r_j,1,1)$ & $\underline{abba}a$ & $(r_j,2,2)$ \\ \hline
			3 & $\underline{ab}ab$ & $(r_j,1,1,r_{j+3})$ & $\underline{abba}ab$ & $(r_j,2,2,r_{j+3})$ \\ \hline
			4.a & $\underline{ab}ba$ & $(r_j,2,r_{j+2})$ & $\underline{abba}ba$ & $(r_j,2,1,1,r_{j+2})$ \\
			4.b & $\underline{ab}b$ & $(r_j,2)$ & $\underline{abba}b$ & $(r_j,2,1,1)$ \\ \hline
			5 & $\underline{ab}bb$ & $(r_j,r_{j+1})$ & $\underline{abba}bb$ & $(r_j,2,1,r_{j+1}-1)$ \\
		\end{tabular}
		\caption{Duplication constellations for $\ell=2$}
		\label{tab:Duplication_Constellations}
	\end{table}
	In Table \ref{tab:Duplication_Constellations}, $a \in \mathbb{Z}_2$ and $b \in \mathbb{Z}_2$ denote two distinct symbols with $a\neq b$, $j \in \mathbb{N}$ denotes the run in which the duplication occurred and $r_j, r_{j+1}, r_{j+2}, r_{j+3}$ the lengths of the $j$-th, $(j+1)$st, $(j+2)$nd, respectively $(j+3)$rd run. Note that the cases 1.b, 2.b, and 4.b refer to the case when the palindromic duplication occurred one symbol before the ending of word. It can be observed that for each case, the number of length-$1$ runs is changed by a distinct value. In case $1$, the number of length $1$ runs is not changed, in case $2$, it decreases by $1$, in case $3$, it decreases by $2$, in case $4$ it increases by $2$ and in the last case it increases by $1$. This enables to distinguish between these cases, if we choose codewords from the set of words that satisfy the run-profile constraint $\mathcal{D}_a(n)$. Assume, the word $\ve{y}$ is received. The decoder then computes $\left\langle a - r^{(1)}(\ve{y}) \right\rangle_5$, which allows to identify one of the above five cases. \\
	We will now show that, given an erroneous channel output $\ve{y} \in \mathbb{Z}_2^{n+2}$, for each of the above cases, we can unambiguously determine the run $j \in \{1, 2, \dots, r(\ve{x}) \}$ in which the palindromic duplication occurred by computing the checksum difference $\tilde{C} = \left\langle C(\ve{y}) - b \right\rangle_{2n+1}$. 
	\begin{table}[h]
		\centering
		\begin{tabular}{l|c|c|c}
			Case & $\tilde{C}$ & Range of $j$ & Range of $r(\ve{y})$   \\ \hline
			1.a & $2j$ & $1\leq j \leq r(\ve{y})$ & $1\leq r(\ve{y}) \leq n-1$  \\
			1.b & $2r(\ve{y})-1$ & -- & $3\leq r(\ve{y}) \leq n+1$  \\ \hline
			2.a & $2j+3$ & $1 \leq j \leq r(\ve{y})-2$ & $3\leq r(\ve{y}) \leq n-1$  \\
			2.b & $2r(\ve{y})-1$ & -- & $3\leq r(\ve{y}) \leq n$ \\ \hline
			3 & $2j+3$ & $1 \leq j \leq r(\ve{y})-3 $ & $4\leq r(\ve{y}) \leq n$ \\ \hline
			4.a & $2j+5 + 2\sum_{k=j+4}^{r(\ve{y})} r_k(\ve{y})$ & $1 \leq j \leq r(\ve{y}) - 4$ & $5\leq r(\ve{y}) \leq n+1 $ \\
			4.b & $ 2r(\ve{y}) - 1$ & -- & $4\leq r(\ve{y}) \leq n + 1$ \\ \hline
			5 & $2j+3 + 2\sum_{k=j+3}^{r(\ve{y})} r_k(\ve{y})$ & $1 \leq j \leq r(\ve{y})-3 $ & $4\leq r(\ve{y}) \leq n$
		\end{tabular}
		\caption{Checksum differences for all cases}
		\label{tab:Checksum_Increases}
	\end{table}
	Table \ref{tab:Checksum_Increases} shows the increases of the checksum for the five different cases. To begin with, we show that the cases 1.a and 1.b, respectively 2.a and 2.b or 4.a and 4.b can be distinguished using $\tilde{C}$. Case 1.a and 1.b can be distinguished since 1.a yields even and 1.b yields odd integers or $0$ for $\tilde{C}$. Cases 2.a and 2.b have the same checksum only for $j=r(\ve{y})-2$, which in both cases means, that the duplication was in the run $r(\ve{y})-2$. Cases 4.a and 4.b only give the same checksum $ 2r(\ve{y}) - 1$, if there is an alternating sequence after the palindromic duplication in case 4.a, which corresponds to case 4.b. \\ Having available the exact case, the run $j$, in which the palindromic duplication occurred, is then obtained by finding the index $j$, which gives the observed checksum deficiency. In cases $4$ and $5$ there might be several adjacent $j$, which satisfy the equation. As this might only occur for sequences, which are alternating before the palindromic duplication, we can identify the position of the palindromic duplication by identifying the run $j$, which satisfies the checksum difference $\tilde{C}$ and ends in a palindromic duplication.
	\qed
\end{proof}
We illustrate the decoding process with the following example.
\begin{example}[Correcting a palindromic duplication]
	Consider the sequence $\ve{x} = (01011001) \in \mathcal{C}_{4, 13}(8)$, which is transmitted over a channel and results in the received word $\ve{y} = (01\underline{0110}1001)$, where the underlined part is the part, which is palindromic duplicated. We have $r(\ve{x}) = (1,1,1,2,2,1)$ and $C(\ve{x}) = 1+2+3+2\cdot 4 + 2\cdot 5 + 6 = 30 = 13 \mod 17$. Thus, $\ve{x} \in \mathcal{D}_4(8) \cap \mathcal{E}_{13}(8)$. For the received word $\ve{y}$, we have $r(\ve{y}) = (1,1,1,2,1,1,2,1)$, $r^{(1)}(\ve{y}) = 6 = 1 \mod 5$ and $C(\ve{y}) = 1+2+3+2\cdot 4 +  5 + 6 + 2\cdot 7 + 8 = 47 = 13 \mod 17$. This means the number of run of length $1$ increased by $\left\langle1-4\right\rangle_5 = 2$, which means we are in case 4. Let us now find those values $j$, for which $2j+5 + 2\sum_{k=j+4}^{r(\ve{y})} r_k(\ve{y}) \overset{!}{=} \tilde{C} = 0 \mod 17$. A quick calculation yields two candidates $j=2$ and $j=3$, from which only $j=3$ is possible, since this is the only run, which ends in a palindrome. Deleting the palindrome which starts at the third run gives the correct transmit word $\ve{x}$.
\end{example}
\begin{corollary}
	There exist $a \in \{0,1,2,3,4\}$ and $b \in \{0,1,\dots,2n\}$ such that
	\[|\mathcal{C}_2(n)| \geq \frac{2^n}{5(2n+1)}.\]
	
\end{corollary}
\begin{proof}
	With the pigeonhole principle, we can find $a \in \{0,1,2,3,4\}$ and $b \in \{0,1,\dots,2n\}$ that yield a code size of at least $\frac{2^n}{5(2n+1)}$. \qed
\end{proof}
The redundancy of Construction \ref{con:palindrome:2} is at most $\log_2(n) + \log_2(10)$ while the best known burst insertion (length $2$) correcting code \cite{Schoeny17} has redundancy at most $\log_2(n) + \log_2(\log_2(n)) +1$. The redundancy scaling of the proposed construction is therefore lower by the term $\log_2(\log_2(n))$.
\subsection{Construction with Palindrome-Free Strings}
In contrast to the previous sections, where we discussed duplications of a fixed length $\ell$, we are now considering single palindromic duplications of arbitrary lengths $2 \leq \ell \leq n$. Note that the case $\ell=1$ is excluded here, since in this case, a palindromic duplication is a single duplication error, which has been studied in, e.g. \cite{Dolecek10,Levenshtein65}. Due to the fact that received words can have length $n$ up to $2n$, a sphere packing approach will not yield good lower bounds for this error model. In the following, we propose a code, which corrects a single palindromic duplication of any length $2$ up to $n$ by using palindrome free words.
\begin{definition}
	A word $\ve{x} \in \Zq^n$ is called \emph{$\ell$-palindrome free}, if $S^{\rho^D_\ell}_1(\ve{x}) = \emptyset$.
\end{definition}
\begin{example}
	The word $\ve{x} = (012122)$ is $2$-palindrome free, while the word $\ve{y} = (012212)$ is \emph{not} $2$-palindrome free, as it contains $(1221)$, which is a palindrome of length $2$.
\end{example}
Due to the combinatorial structure of palindromes, it is intuitive, that any word that is $\ell$-palindrome free, also does not contain any palindrome of length at least $\ell$. This is shown in them following lemma.
\begin{lemma}
	For $\ve{x} \in \Zq^n$ and $\ell_2 \geq \ell_1$, we have
	\[S^{\rho^D_{\ell_1}}_1(\ve{x}) = \emptyset \Longrightarrow S^{\rho^D_{\ell_2}}_1(\ve{x}) = \emptyset. \]
\end{lemma}
\begin{proof}
	By the definition of a $\ell_1$-palindrome-free word, we see that $S^{\rho^D_{\ell_1}}_1(\ve{x}) = \emptyset \Longleftrightarrow \nexists \, 1 \leq i \leq n-2\ell_1 : x_{i-j+\ell_1} = x_{i+j+1+\ell_1} \, \forall \, j \in \{0,1, \dots \ell_1-1\}$. Since $\ell_2 \geq \ell_1$, it follows that $\nexists \, 1 \leq i \leq n-2\ell_2 : x_{i-j+\ell_2} = x_{i+j+1+\ell_2} \, \forall \, j \in \{0,1, \dots \ell_2-1\}$. \qed
\end{proof}
Now for our construction, we consider words, that do not contain any palindrome of length $2$, i.e. $S^{\rho^D_2}(\ve{x}) = \emptyset$. For these words, the following lemma holds.
\begin{lemma} \label{lemma:2:palindrome:free:words}
	Let $\ve{x}, \ve{y} \in \Zq^n$ be two $2$-palindrome free words, i.e. $S^{\rho^D_2}_1(\ve{x}) = S^{\rho^D_2}_1(\ve{y}) = \emptyset$. Then
	\[S^{\rho_\ell}_1(\ve{x}) \cap S^{\rho_\ell}_1(\ve{y}) = \emptyset,\]
	for all $\ell \geq 2$.
\end{lemma}
\begin{proof}
	Consider the system of equations $\rho_{\ell}(\ve{x},i) = \rho_{\ell}(\ve{y},i+j)$ in Appendix \ref{sec:conditions:equivalence:palindromic:duplications:two:words}. Plugging in, \eqref{eq:xy:cond:j:less:l:a} into \eqref{eq:xy:cond:j:less:l:b}, we yield $y_{i+\ell-m} = y_{i+\ell+1+m}$ for $m \in \{0, \dots, j-1\}$, which corresponds to a palindrome of length $j$. Since neither $\ve{x}$ nor $\ve{y}$ can have a palindrome of length greater than $2$, it follows that $\rho_{\ell}(\ve{x},i) = \rho_{\ell}(\ve{y},i+j)$ cannot hold for any $i$ and $j<\ell$. The same can be shown for $j \geq \ell$, since in this case \eqref{eq:xy:cond:j:geq:l:a} and \eqref{eq:xy:cond:j:geq:l:b} imply a palindrome of length $\ell$ in either $\ve{x}$ or $\ve{y}$. \qed
\end{proof}
We therefore construct the following code, which can correct a single palindromic duplication of length $2$ up to $n$.
\begin{construction} \label{con:palindrome:free:words} We construct the code $\mathcal{C}_{\mathrm{PF}}(n)$, consisting of all $2$-palindrome free words as
	\[ \mathcal{C}_{\mathrm{PF}}(n) = \{ \ve{x} \in \Zq^n : S^{\rho^D_2}_1(\ve{x}) = \emptyset \}. \]
\end{construction}
With Lemma \ref{lemma:2:palindrome:free:words}, the code $\mathcal{C}_{\mathrm{PF}}$ can correct a single palindromic duplication of length $2$ to $n$. In the following, we investigate the cardinality and rate of Construction \ref{con:palindrome:free:words} and their asymptotic behavior.
\begin{lemma}\setlength{\arraycolsep}{4pt} \label{lemma:cardinality:palindrome:free}
	The cardinality of Construction \ref{con:palindrome:free:words} is
	\[ |\mathcal{C}_{\mathrm{PF}}(n)| = \sum_{i=1}^3 c_i(q) \lambda^{n-3}_i(q) , \]
	where
	\[ c_i(q) = q(q-1)\frac{(q^2+q)\lambda_i^2(q) + (q^2-1)\lambda_i(q) + q^2}{(q-1)\lambda_i^2(q) + (2q-4)\lambda_i(q) + 3q-3}, \]
	and $\lambda_i(q)$ are the solutions to the polynomial equation
	\[ -\lambda^3 + (q-1)\lambda^2 + (q-2)\lambda + (q-1) = 0. \]
\end{lemma}
\begin{proof}
	\setlength{\arraycolsep}{4pt}
	By the definition of the code $\mathcal{C}_{\mathrm{PF}}$, the cardinality $|\mathcal{C}_{\mathrm{PF}}|$ is given by the number of words $\ve{x} \in \Zq^n$ that do not contain a palindrome of length $2$. We present a recursive approach on how to compute this number. Let
	\[\ve{\mu}(n) = \begin{pmatrix} M_{aaa}(n) & M_{aab}(n) & M_{aba}(n) & M_{abb}(n) & M_{abc}(n) \end{pmatrix} \in \mathbb{N}^5 \]
	be a vector over the natural numbers, whose entries are defined as
	\begin{align*}
	M_{aaa}(n) &= |\{\ve{x} \in \Zq^n : S^{\rho^D_2}_1(\ve{x}) = \emptyset \land x_{n-2}=x_{n-1}=x_n \}|, \\
	M_{aab}(n) &= |\{\ve{x} \in \Zq^n : S^{\rho^D_2}_1(\ve{x}) = \emptyset \land x_{n-2}=x_{n-1}\neq x_n \}|, \\
	M_{aba}(n) &= |\{\ve{x} \in \Zq^n : S^{\rho^D_2}_1(\ve{x}) = \emptyset \land x_{n-2}=x_n\neq x_{n-1} \}|, \\
	M_{abb}(n) &= |\{\ve{x} \in \Zq^n : S^{\rho^D_2}_1(\ve{x}) = \emptyset \land x_{n-2}\neq x_{n-1}=x_n \}|, \\
	M_{abc}(n) &= |\{\ve{x} \in \Zq^n : S^{\rho^D_2}_1(\ve{x}) = \emptyset \land x_{n-2}\neq x_{n-1}, x_{n-1}\neq x_n, x_{n-2}\neq x_n \}|.
	\end{align*}
	With this definition, the elements of $\ve{\mu}(n)$ count the number of $2$-palindrome free words, which end with a specific pattern that is indicated by subscripts. A recursive relation $\ve{\mu}(n+1) = \ve{A} \ve{\mu}(n)$ can be found by counting the number of words of length $n+1$ with a specific ending pattern that are obtained by adding a symbol to a word of length $n$ with another ending pattern, such that the resulting word of length $n+1$ is still $2$-palindrome free. This recursive relation is illustrated at the exemplary case of $M_{aab}(n+1)$ in the following and can be deduced for the other cases in a similar fashion. Words of length $n+1$ that end on a pattern $aab$, where $a,b \in \Zq$ and $a\neq b$ can be created by adding a symbol $b$ to a word that ends on a pattern $aa$, i.e. $x_{n-1} = x_n$. Since the resulting word is not allowed to contain a palindrome of length $2$, we can either
	\begin{itemize}
		\item append $b\neq a$ to the pattern $aaa$
		\item or append $c$ with $c \neq a$ and $c \neq b$ to the pattern $abb$.
	\end{itemize}
	For the first case, there are $q-1$ possibilities to choose $b$ and for the second case, there are $q-2$ possibilities to choose $c$. Therefore, we obtain
	\begin{align*}
	M_{aab}(n+1) = (q-1) M_{aaa}(n) + (q-2) M_{abb}(n).
	\end{align*}
	Repeating the same steps for the other patterns, we yield the linear recursion $\ve{\mu}(n+1) = \ve{A} \ve{\mu}(n)$, where $\ve{A} \in \mathbb{N}_0^{5 \times 5}$ is given by
	\[\ve{A} = \begin{pmatrix}
	\,\,0\,\, & q-1 & \,\,0\,\,   & \,\,0\,\, & 0 \\
	0 & 0   & 1   & 1 & q-2\\
	0 & 0   & 1   & 1 & q-2\\
	1 & q-2 & 0   & 0 & 0\\
	0 & 0   & 1   & 1 & q-2 
	\end{pmatrix}.\]
	The starting conditions for the recursion are $M_{aaa}(3) = q$, $M_{aab}(3) = M_{aba}(3) = M_{abb}(3) = q(q-1)$ and $M_{abc}(3) = q(q-1)(q-2)$, since words of length $3$ cannot contain any palindrome of length $2$. Solving the recursion with standard techniques yields the lemma. \qed
\end{proof}
\begin{lemma} \label{lemma:asymptotic:rate}
	The rate $R_{\mathrm{PF}} = \frac{\log_q(|\mathcal{C}_{\mathrm{PF}}(n)|)}{n}$ of Construction \ref{con:palindrome:free:words} tends to $R_{\mathrm{PF}} \rightarrow \log_q (\lambda(q))$, where
	\begin{align}
	\lambda(q) &= \frac{q-1}{3}+ \sqrt[3]{A + \sqrt{A^2 - B^3}} + \sqrt[3]{A - \sqrt{A^2 - B^3}}, \nonumber \\
	A &= \frac{q-1}{2} + \frac{(q - 1)(q - 2)}{6} + \frac{(q-1)^3}{27}, \label{eq:largest:ev:palindromes} \\
	B &= \frac{q-2}{3} + \frac{(q-1)^2}{9}, \nonumber
	\end{align}
	as $n \rightarrow \infty$.
\end{lemma}
\begin{proof}
	From Lemma \ref{lemma:cardinality:palindrome:free}, we know that $|\mathcal{C}_{\mathrm{PF}}(n)| = \sum_{i=1}^3 c_i(q) \lambda^{n-3}_i(q)$, where $\lambda_i(q)$ are the solutions to $-\lambda^3+(q-1)\lambda
	^2+(q-2)\lambda+q-1 = 0$. For large $n$, the cardinality of $\mathcal{C}_{\mathrm{PF}}(n)$ is therefore dominated by the largest root $\lambda(q)$ and thus
	\[\frac{\log_q\left(|\mathcal{C}_{\mathrm{PF}}(n)|\right)}{\log_q\left(\lambda^n(q)\right)} \rightarrow 1, \] 
	as $n \rightarrow \infty$. Hence, we can say that asymptotically $R_{\mathrm{PF}} \rightarrow \log_q(\lambda(q))$ as $n \rightarrow \infty$. $\lambda(q)$ is given explicitly by \eqref{eq:largest:ev:palindromes}. \qed
\end{proof}
On the other hand, the following property of the code rate $R_{\mathrm{PF}}$ can be shown for large alphabet sizes $q$.
\begin{lemma} \label{lemma:asymptotic:rate:q}
	For every code length $n$, the rate of Construction \ref{con:palindrome:free:words} satisfies the asymptotic property $\lim\limits_{q\rightarrow \infty} R_{\mathrm{PF}} =  1$.
\end{lemma}
\begin{proof}
	In Lemma \ref{lemma:cardinality:palindrome:free}, we have seen that $|\mathcal{C}_{\mathrm{PF}}(n)| = \sum_{i=1}^3 c_i(q) \lambda^{n-3}_i(q)$. Analyzing \eqref{eq:largest:ev:palindromes} and with the expressions for $c(q)$ for the coefficient corresponding to the largest root $\lambda(q)$ we yield
	\begin{align*}
	\lim\limits_{q\rightarrow \infty} \frac{\lambda(q)}{q} &= 1,\\
	\lim\limits_{q \rightarrow \infty} \frac{c(q)}{q^3} &= 1.
	\end{align*}
	By standard methods, it can further be shown that the other two roots converge to a constant $\lim\limits_{q\rightarrow \infty} \lambda_i(q) = -\frac12 \pm \frac{\sqrt{3}}{2} \mathrm{j}$. Hence,
	\[\lim\limits_{q\rightarrow \infty} \frac{|\mathcal{C}_{\mathrm{PF}}(n)|}{q^n} = \lim\limits_{q\rightarrow \infty} \sum_{i=1}^3 \frac{c_i(q)}{q^3} \frac{\lambda^{n-3}_i(q)}{q^{n-3}} = 1. \]
	\qed
\end{proof}
Table \ref{tab:palindrome:free:rates} summarizes rates of $\mathcal{C}_{\mathrm{PF}}$ for different code lengths $n$ and alphabet sizes $q$.
\begin{table}[h]
	\centering
	\begin{tabular}{c|ccccccccc}
		\diagbox{$q$}{$n$} & $2$ & $4$ & $8$ & $16$ & $32$ & $64$ & $128$ & $256$ & $\infty$ \\ \hline
		$2$ & $1$ & $0.896$ & $0.792$ & $0.639$ & $0.595$ & $0.573$ & $0.562$ & $0.557$ & $0.552$ \\
		$3$ & $1$ & $0.973$ & $0.932$ & $0.911$ & $0.901$ & $0.895$ & $0.893$ & $0.892$ & $0.892$ \\
		$4$ & $1$ & $0.988$ & $0.971$ & $0.962$ & $0.957$ & $0.955$ & $0.954$ & $0.954$ & $0.953$ \\
		$5$ & $1$ & $0.994$ & $0.984$ & $0.979$ & $0.977$ & $0.976$ & $0.975$ & $0.975$ & $0.975$ \\
		$\infty$ & $1$ & $1$ & $1$ & $1$ & $1$ & $1$ & $1$ & $1$ & $1$ \\
	\end{tabular}
	\caption{Rates of $\mathcal{C}_{\mathrm{PF}}$ for different code lengths $n$ and alphabet sizes $q$}
	\label{tab:palindrome:free:rates}
\end{table}
By Lemma \ref{lemma:asymptotic:rate} the last column of Table \ref{tab:palindrome:free:rates} are the values $\log_q(\lambda(q))$ and by Lemma \ref{lemma:asymptotic:rate:q} the last row of Table \ref{tab:palindrome:free:rates} is equal ot $1$. It is observed that already for moderate $q$, the rate of this construction is close to $1$. However, for large $n$, the rate converges to $\log_q(\lambda(q)) < 1$. Also note that for $n\leq 3$, $R_{\mathrm{PF}} = 1$ for any $q$, since words of length at most $3$ are automatically $2$-palindrome free. \\

Decoding a received word $\ve{y} \in \mathbb{Z}_q^{n+\ell}$ can be done by first identifying the length of the duplication $\ell$ from the length of the received word. Then, the transmitted word can be found by deleting palindromic duplication at every possible position in $\ve{y}$ and deciding for the result, which gave a $2$-palindrome free word.
\subsection{Comparison with Burst Insertion Correcting Codes}
Figure \ref{fig:bounds_vs_burst} shows the lower bounds (LB) on the redundancy for binary codes and different duplication lengths $\ell$. We compare our results with maximum redundancies of single burst insertion correcting codes from \cite{Schoeny17}. To the best of our knowledge, these constructions have the largest codebooks that can correct a single burst insertion of length $\ell$. The figure also includes the redundancies from the single tandem and palindromic duplication correcting construction, presented in Section \ref{ss:code:single:tandem:duplication} ($\mathcal{C}_1(n)$), denoted by \textit{VT Tan.} and \ref{ss:code:single:palindromic:duplication} ($\mathcal{C}_2(n)$), denoted by \textit{Pal.}. Interestingly, there is a significant gap between the redundancies of existing burst insertion constructions and the derived lower bounds on the redundancy.
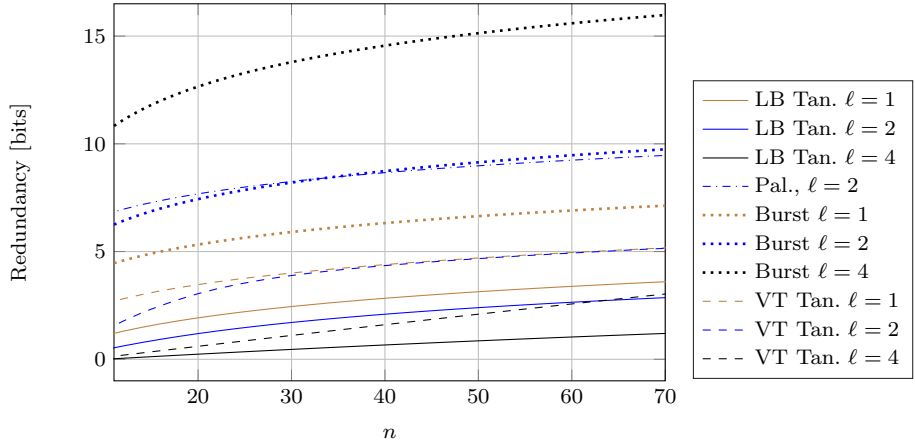
\begin{figure}[htp]
	\centering
	\input{Tan_Pal_Burst.tikz}
	\caption{Tandem/palindromic duplication bounds vs. burst insertion redundancies}
	\label{fig:bounds_vs_burst}
\end{figure}
\section{Conclusion \& Outlook}
In this paper we have derived upper bounds on the cardinalities of codes correcting tandem or palindromic duplication errors of a given length $\ell$. We have derived constructions that correct a single tandem or palindromic duplication and seen that these construction yield lower redundancies than codes that correct an arbitrary burst of insertions. However, there remain several interesting aspects in this field, as
\begin{itemize}
	\item Asymptotic behavior of code size upper bounds
	\item Code size upper bounds for multiple palindromic duplications
	\item Code constructions correcting multiple tandem or palindromic duplications
	\item Code constructions correcting a combination of duplication errors and other error types, such as substitution errors or insertion and deletion errors,
\end{itemize}
and many more.
\appendix
\section{Sphere Sizes for Tandem and Palindromic Duplications and Deletions} \label{sec:error:spheres:tandem:palindromic:dupliations:deletions}
In the following we derive the size of the spheres $S^{\epsilon}_t(\ve{x})$, as defined in \eqref{eq:error:sphere}, for tandem and palindromic duplication and deletion errors. For the subsequent two lemmas we denote the $\ell$-step derivative by $\phi_\ell(\ve{x})=(\ve{u}_x, \ve{v}_x)$, according to the definition from Section \ref{ss:equivalence:tandem:duplication:deletion:codes}.
\begin{lemma}
	The sphere size for tandem duplications of length $\ell$ is given as
	\begin{equation*}
	|S^{\tau_\ell}_{t}(\ve{x})| = \binom{wt_{\mathrm{H}}(\ve{v}_x)+t}{t},
	\end{equation*}
	where $\phi_\ell(\ve{x})=(\ve{u}_x, \ve{v}_x)$ is the $\ell$-step derivative of $\ve{x}$.
\end{lemma}
\begin{proof}
	Recall that a tandem duplication error corresponds to increasing one entry of the $\ell$-zero signature $\sigma_\ell(\ve{v}_x)$ by one. Then, the duplication sphere size equals the number of vectors $\ve{y} \in \mathbb{N}_0^{|\sigma_\ell(\ve{v}_x)|}$ with $\ve{y} \geq \sigma_\ell(\ve{v}_x)$ and $|\ve{y}|_1=|\sigma_\ell(\ve{v}_x)|_1 + t$. The number of such vectors is given by $\binom{|\sigma_\ell(\ve{v}_x)|+t-1}{t} = \binom{wt_{\mathrm{H}}(\ve{v}_x)+t}{t}$ \qed
\end{proof}
\begin{lemma}\label{lemma:tandem:deletion:sphere}
	The sphere size for tandem deletions of length $\ell$ is given as
	\begin{align*}
		|S^{\tau_\ell^D}_t(\ve{x})| &= |\{ \ve{s} \in \mathbb{N}_0^{|\sigma_\ell(\ve{v}_x)|} : \ve{s} \leq \sigma_\ell(\ve{v}) \land |\ve{s}|_1 = |\sigma_\ell(\ve{v})|_1 - t \}| \\
		&= |\{ \ve{s} \in \mathbb{N}_0^{|\sigma_\ell(\ve{v}_x)|} : \ve{s} \leq \sigma_\ell(\ve{v}) \land |\ve{s}|_1 = t \}|,
	\end{align*}
	where $\phi_\ell(\ve{x})=(\ve{u}_x, \ve{v}_x)$ is the $\ell$-step derivative of $\ve{x}$.
\end{lemma}
\begin{proof}
	A tandem deletion corresponds to decreasing one entry of the $\ell$-zero signature $\sigma_\ell(\ve{v}_x)$ by one. It is only possible to delete a tandem duplication at positions, where the $\ell$-zero signature has positive entries. \qed
\end{proof}
Note that by this Lemma, $S^{\tau_\ell^D}_t(\ve{x}) = \emptyset$, if $|\sigma_\ell(\ve{v})|_1 < t$.
\begin{corollary}\label{cor:single:tandem:deletion:sphere}
	The sphere size for single tandem deletions of length $\ell$ is
	\begin{equation*}
	|S^{\tau_\ell^D}_1(\ve{x})| = wt_{\mathrm{H}}(\sigma_\ell(\ve{v})),
	\end{equation*}
	where $\phi_\ell(\ve{x})=(\ve{u}, \ve{v})$ is the $\ell$-step derivative of $\ve{x}$.
\end{corollary}
We continue with deriving the palindromic duplication sphere size for the cases $\ell=1$ and $\ell=2$. For $\ell=1$, a palindromic duplication is a single duplication. Therefore, the sphere size is
\begin{equation*}
|S^{\rho_1}_1(\ve{x})| = r(\ve{x}),
\end{equation*}
as duplications in the same run yield the same outcome.
\begin{lemma}\label{lemma:palindromic:duplication:sphere:l:2}The size of the palindromic duplication sphere $|S^{\rho_2}_1(\ve{x})|$ for palindromic duplications of length $2$ is
	\begin{equation*}
	|S^{\rho_2}_1(\ve{x})| = n-1 - \sum_{i=3}^n (i-2) r^{(i)}(\ve{x}) = 2 r(\boldsymbol{x}) - r^{(1)}(\boldsymbol{x})-1.
	\end{equation*}
\end{lemma}
\begin{proof}
	We start with the observation that there are $n-1$ possible positions $i \in \{0,1, ..., n-2\}$ for palindromic duplications. Now, for $\ell=2$, the conditions $\rho_{\ell}(\ve{x},i) = \rho_{\ell}(\ve{x},i+j)$ \eqref{eq:cond:j:less:l:a} - \eqref{eq:cond:j:less:l:c} and \eqref{eq:cond:j:geq:l:a} - \eqref{eq:cond:j:geq:l:c} become  $x_1 = x_2 = \dots = x_{2+j} \, \forall \, j>0$. We therefore deduce that two palindromic duplications in $\ve{x}$ of length $2$ only result in the same vector $\boldsymbol{y}=\rho_{\ell}(\ve{x},i) = \rho_{\ell}(\ve{x},i+j)$ iff they appear in the same run in $\ve{x}$. Further, two palindromic duplications at two different positions $i$ and $i+j, j>0$ can only duplicate symbols from the same run, if this run has length at least $3$. Thus, every additional symbol to runs of length at least $2$ does not increase the duplication sphere size and has to be subtracted from the palindromic duplication sphere size. Using $\sum_{i=1}^n i r^{(i)}(\ve{x}) = n$ and $\sum_{i=1}^n r^{(i)}(\ve{x}) = r(\ve{x})$ yields the statement. \qed
\end{proof}
For $\ell \geq 3$ and $j \geq 2$, \eqref{eq:cond:j:less:l:a} - \eqref{eq:cond:j:less:l:c} and \eqref{eq:cond:j:geq:l:a} - \eqref{eq:cond:j:geq:l:c} do not imply $x_1=x_2 = \dots =x_{\ell+j}$. For example, consider $\ell=3$ and the word $\ve{x} = (010010)$. Then, $\rho_{3}(\ve{x},0) = \rho_{3}(\ve{x},3) = (010010010)$. However, it is possible to find an upper bound on the size of the palindromic duplication sphere. For $j = 1$, \eqref{eq:cond:j:less:l:a} - \eqref{eq:cond:j:less:l:c} become $x_1 = x_2 = \dots = x_{\ell+1}$. Therefore two neighboring palindromic duplications can only result in the same word if they appear in one run.
\begin{lemma} The size of the palindromic duplication sphere $S^{\rho_\ell}_1(\ve{x})$ is upper bounded by
	\begin{equation*}
	|S^{\rho_\ell}_1(\ve{x})| \leq n - \ell +1 - \sum_{i=\ell+1}^{n} (i-\ell)r^{(i)}(\ve{x}).
	\end{equation*}
\end{lemma}
\begin{proof}
	There are $n-\ell+1$ possible positions for palindromic duplications of length $\ell$. Now, as seen before, duplications in the same run result in the same descendant. We therefore subtract the additional $i-\ell$ entries of runs with length at least $\ell+1$ from the number of possible positions for duplications to obtain an upper bound on the duplication sphere.  \qed
\end{proof}
Similar to the previous discussion, we start with deriving the size of the palindromic deletions spheres for $\ell=1$ and $\ell=2$. For $\ell=1$, a palindromic deletion is a de-duplication of one symbol. Therefore, the size of the error sphere becomes
\begin{equation}
|S^{\rho_1^D}_1(\ve{x})| = r^{(\geq 2)}(\ve{x}) ,
\end{equation}
where $r^{(\geq 2)}(\ve{x})$ is the number of runs of length at least $2$. Further, we derive the following lemma for binary words.
\begin{lemma} \label{lemma:palindromic_deletion_sphere_l_2}
	The size of the palindromic deletion sphere $|S^{\rho^D_2}_1(\ve{x})|$ for $q=2$ is
	\begin{equation*}
	|S^{\rho^D_2}_1(\ve{x})| = r^{(2)}_{\mathcal{I}}(\ve{x}) + r^{(\geq 4)}(\ve{x}),
	\end{equation*}
	where $r^{(2)}_{\mathcal{I}}(\ve{x})$ is the number of runs of length $2$, that are located at the interior of $\ve{x}$, i.e., between $x_2$ and $x_{n-1}$ and, $r^{(\geq 4)}(\ve{x})$ denotes the number of runs of length at least $4$ in $\ve{x}$.
\end{lemma}
\begin{proof}
	There are $4$ possible patterns $(0000)$, $(1111)$, $(0110)$, $(1001)$, at which palindromic deletions of length $2$ can occur. Recall that, as we have seen in the proof of Lemma \ref{lemma:palindromic:duplication:sphere:l:2}, two palindromic deletions of length $2$ at two distinct positions in a word $\ve{x}$ can only results in the same outcome, if they appear in the same run. Every run of length at least $4$ contains one of the patterns $(0000)$, $(1111)$ and therefore will contribute one element to the palindromic deletion sphere. The patterns $(0110)$, $(1001)$ contain a run of length exactly $2$, that is located in the interior of $\ve{x}$, such that there is at least one symbol to the left and right of the run. Thus, every run of length $2$, that is located in the interior of $\ve{x}$ also contributes one unique element in the palindromic deletion sphere. Therefore, the total size of the deletion sphere is $r^{(2)}_{\mathcal{I}}(\ve{x}) + r^{(\geq4)}(\ve{x})$.  \qed
\end{proof}
Let us define the matrix $\ve{A}^{\rho_\ell}(\ve{x}) \in \mathbb{Z}_q^{\ell\times n-2\ell+1}$ to be
\begin{equation}
\ve{A}^{\rho_\ell}(\ve{x}) = \begin{pmatrix}
x_{2\ell}-x_1 & x_{2\ell+1}-x_2 & \dots & x_n-x_{n-2\ell+1} \\
x_{2\ell-1}-x_2 & x_{2\ell}-x_3 & \dots & x_{n-1}-x_{n-2\ell+2} \\
\vdots & \vdots &\ddots & \vdots \\
x_{\ell+1}-x_\ell & x_{\ell+2}-x_{\ell+1} & \dots & x_{n-\ell+1}-x_{n-\ell}
\end{pmatrix}.
\end{equation}
With this definition it is directly possible to establish the following upper bound on the size of the palindromic deletion spheres for arbitrary deletion length $\ell$.
\begin{lemma} \label{lemma:upper:bound:deletion:sphere}
	The palindromic deletion sphere $|S^{\rho^D_\ell}_1(\ve{x})|$ is upper bounded by
	\begin{equation*}
	|S^{\rho^D_\ell}_1(\ve{x})| \leq r^{(0)} \left(\ve{A}^{\rho_\ell}(\ve{x})\right),
	\end{equation*}
	where $r^{(0)} \left(\ve{A}^{\rho_\ell}(\ve{x})\right)$ is the number of runs of all zero columns in $\left(\ve{A}^{\rho_\ell}(\ve{x})\right)$.
\end{lemma}
\begin{proof}
	A palindrome of length $\ell$ in the word $\ve{x}$ corresponds to a zero column in the matrix $\ve{A}^{\rho_\ell}(\ve{x})$. Therefore palindromic deletions are only possible at positions $i$, where $\ve{A}^{\rho_\ell}(\ve{x})$ has a zero-column. Further, it can be shown that two neighboring zero columns are only possible if $x_{i+1} = x_{i+2} = \dots = x_{i+2\ell+1}$, i.e. for a run of length $2\ell+1$. However, two palindromic deletions inside the same run result in the same words. Therefore, every run of all zero columns in $\left(\ve{A}^{\rho_\ell}(\ve{x})\right)$ contributes one unique element to $S^{\rho^D_\ell}_1(\ve{x})$. \qed
\end{proof}
\begin{example}
	Consider the word $\ve{x} = (21011012210) \in \mathbb{Z}_3^{11}$. The palindromic deletion sphere for deletions of length $3$ is given by $S^{\rho^D_3}_1(\ve{x}) = \{ (21012210), (21011012) \}$. The matrix $\ve{A}^{\rho_3}(\ve{x})$ is given by
	\[
	\ve{A}^{\rho_3}(\ve{x}) = \begin{pmatrix} 1 & 0 & 2 & 1 & 0 & 0 \\
	0 & 0 & 0 & 1 & 2 & 0 \\
	1 & 0 & 2 & 1 & 1 & 0
	\end{pmatrix}.
	\]
	Applying Lemma \ref{lemma:upper:bound:deletion:sphere}, yields $|S^{\rho^D_\ell}_1(\ve{x})| \leq 2$.
\end{example}
\section{Equivalence of Palindromic Duplications Errors in One Word} \label{sec:conditions:equivalence:palindromic:duplications}
In this section we derive conditions that two palindromic duplications, respectively deletions at two different positions $i$ and $i+j$ with $j > 0$ result in the same word $\rho_{\ell}(\ve{x},i) = \rho_{\ell}(\ve{x},i+j)$, respectively $\rho^D_{\ell}(\ve{x},i) = \rho^D_{\ell}(\ve{x},i+j)$ for palindromic deletions. For $j < \ell$ the condition $\rho_{\ell}(\ve{x},i) = \rho_{\ell}(\ve{x},i+j)$ can be expressed as (the left hand side of the equations corresponds to $\rho_{\ell}(\ve{x},i+j)$ and the right hand side to $\rho_{\ell}(\ve{x},i)$)\\
\vspace{-15pt}
\begin{subequations}
	\begin{align}
	x_{i+\ell+1+m}  &= x_{i+\ell-m}, \quad &m &\in \{0, \dots, j-1\}, \label{eq:cond:j:less:l:a} \\
	x_{i+\ell+2j-m}  &= x_{i+\ell-m}, \quad& m &\in \{j, \dots, \ell-1\},\label{eq:cond:j:less:l:b}  \\
	x_{i+\ell+2j-m}  &= x_{i+1+m}, \quad &m &\in \{\ell, \dots, \ell+j-1\}. \label{eq:cond:j:less:l:c}
	\end{align}
\end{subequations}
For $j \geq \ell$ these conditions are \\
\vspace{-15pt}
\begin{subequations}
	\begin{align}
	x_{i+\ell+1+m}  &= x_{i+\ell-m}, \quad &m &\in \{0, \dots, \ell-1\}, \label{eq:cond:j:geq:l:a} \\
	x_{i+\ell+1+m}  &= x_{i+1+m}, \quad &m &\in \{\ell, \dots, j-1\}, \label{eq:cond:j:geq:l:b} \\
	x_{i+\ell+2j-m}  &= x_{i+1+m}, \quad &m &\in \{j, \dots, \ell+j-1\}. \label{eq:cond:j:geq:l:c}
	\end{align}
\end{subequations}
The conditions $\rho^D_{\ell}(\ve{x},i) = \rho^D_{\ell}(\ve{x},i+j)$ for $j>0$ are\\
\vspace{-15pt}
\begin{subequations}
	\begin{align}
	x_{i+\ell+1+m}  &= x_{i+\ell-m}, \quad &m &\in \{0, \dots, \ell-1\}, \label{eq:cond:del:a} \\
	x_{i+\ell+j+1+m}  &= x_{i+\ell+j-m}, \quad &m &\in \{0, \dots, \ell-1\}, \label{eq:cond:del:b} \\
	x_{i+2\ell+1+m}  &= x_{i+\ell+1+m}, \quad &m &\in \{0, \dots, j-1\}. \label{eq:cond:del:c}
	\end{align}
\end{subequations}
\section{Equivalence of Palindromic Duplications in Two Words} \label{sec:conditions:equivalence:palindromic:duplications:two:words}
In this section we derive conditions that two palindromic duplications at two different positions $i$ and $i+j$ with $j > 0$ result in the same word $\rho_{\ell}(\ve{x},i) = \rho_{\ell}(\ve{y},i+j)$. For $j < \ell$ the condition $\rho_{\ell}(\ve{x},i) = \rho_{\ell}(\ve{y},i+j)$ can be expressed as
\begin{subequations}
	\begin{align}
	x_{m} &= y_{m},  &m &\in \{1, \dots, i+\ell\}\cup \{ i+j+\ell+1, \dots, n \}, \label{eq:xy:cond:j:less:l:a} \\
	x_{i+\ell-m} &= y_{i+\ell+1+m},  &m &\in \{0, \dots, j-1\}, \label{eq:xy:cond:j:less:l:b} \\
	x_{i+1+m} &= y_{i+2j+1+m}, & m &\in \{0, \dots, \ell-j-1\},\label{eq:xy:cond:j:less:l:c}  \\
	x_{i+\ell+1+m} &= y_{i+2j-m},  &m &\in \{0, \dots, j-1\}. \label{eq:xy:cond:j:less:l:d}	
	\end{align}
\end{subequations}
For $j \geq \ell$ these conditions are \\
\vspace{-15pt}
\begin{subequations}
	\begin{align}
	x_{m} &= y_{m},  &m &\in \{1, \dots, i+\ell\}\cup \{ i+j+\ell+1, \dots, n \}, \label{eq:xy:cond:j:geq:l:a} \\
	x_{i+\ell-m} &= y_{i+\ell+1+m},  &m &\in \{0, \dots, \ell-1\}, \label{eq:xy:cond:j:geq:l:b} \\
	x_{i+\ell+1+m} &= y_{i+2\ell+1+m}, & m &\in \{0, \dots, j-\ell-1\},\label{eq:xy:cond:j:geq:l:c}  \\
	x_{i+j+1+m} &= y_{i+j+\ell-m},  &m &\in \{0, \dots, \ell-1\}. \label{eq:xy:cond:j:geq:l:d}
	\end{align}
\end{subequations}
The conditions $\rho^D_{\ell}(\ve{x},i) = \rho^D_{\ell}(\ve{y},i+j)$ for $j>0$ are\\
\vspace{-15pt}
\begin{subequations}
	\begin{align}
	x_{m} &= y_{m},  &m &\in \{1, \dots, i+\ell\}\cup \{ i+j+2\ell+1, \dots, n \}, \label{eq:xy:cond:del:a} \\
	x_{i+\ell-m}  &= x_{i+\ell+1+m}, \quad &m &\in \{0, \dots, \ell-1\}, \label{eq:xy:cond:del:b} \\
	y_{i+j+\ell-m}  &= y_{i+j+\ell+1+m}, \quad &m &\in \{0, \dots, \ell-1\}, \label{eq:xy:cond:del:c} \\
	x_{i+2\ell+1+m}  &= y_{i+\ell+1+m}, \quad &m &\in \{0, \dots, j-1\}. \label{eq:xy:cond:del:d}
	\end{align}
\end{subequations}
%
%
%\begin{acknowledgements}
%\end{acknowledgements}

\bibliographystyle{spmpsci}      % mathematics and physical sciences
\bibliography{ref.bib}   % name your BibTeX data base


\end{document}

%% file: Tan_Pal_Burst.tikz
\definecolor{mycolor1}{rgb}{0.00000,0.44700,0.74100}%
\definecolor{mycolor2}{rgb}{0.85000,0.32500,0.09800}%
\definecolor{mycolor3}{rgb}{0.92900,0.69400,0.12500}%
\definecolor{mycolor4}{rgb}{0.49400,0.18400,0.55600}%
\definecolor{mycolor5}{rgb}{0.46600,0.67400,0.18800}%
\definecolor{mycolor6}{rgb}{0.30100,0.74500,0.93300}%
\definecolor{mycolor7}{rgb}{0.63500,0.07800,0.18400}%
\begin{tikzpicture}

\begin{axis}[%
width=7.25cm,
height=5.0cm,
at={(0.758in,0.481in)},
scale only axis,
xmin=11,
xmax=70,
xmajorgrids,
xlabel={$n$},
ymin=-1,
ymax=16.5,
ymajorgrids,
ylabel={Redundancy [bits]},
axis background/.style={fill=white},
legend style={legend cell align=left,align=left,draw=white!15!black,at={(1.25,0.005)},anchor=south},
]
\addplot [color=brown,solid]
  table[row sep=crcr]{%
10	1.0993344187496\\
12	1.30481314442127\\
14	1.48541802168539\\
16	1.64636113114489\\
18	1.79141295484673\\
20	1.92337862351011\\
22	2.04439409786062\\
24	2.15611919700349\\
26	2.25986712562365\\
28	2.35669351278587\\
30	2.44745897691015\\
32	2.53287398955841\\
34	2.61353165291457\\
36	2.68993206032191\\
38	2.76250068627316\\
40	2.83160246830424\\
42	2.89755273102917\\
44	2.96062576166585\\
46	3.02106161552783\\
48	3.0790715709912\\
50	3.13484254228691\\
52	3.18854067961293\\
54	3.24031432939006\\
56	3.29029648632511\\
58	3.33860683870839\\
60	3.38535348604365\\
62	3.43063439302911\\
64	3.47453863665795\\
66	3.51714749949318\\
68	3.55853545269993\\
70	3.59877106988891\\
72	3.63791799648088\\
74	3.6760363820604\\
76	3.7131856178552\\
78	3.7494295319832\\
80	3.7848454141171\\
};
\addlegendentry{LB Tan. $\ell=1$};

\addplot [color=blue,solid]
  table[row sep=crcr]{%
10	0.44477179752932\\
12	0.619719328164265\\
14	0.780549100427336\\
16	0.928710833814446\\
18	1.06553049262317\\
20	1.1922225950491\\
22	1.30988696148736\\
24	1.41950794908536\\
26	1.52195836109268\\
28	1.61800683169757\\
30	1.70832717347472\\
32	1.79350849047715\\
34	1.87406523840163\\
36	1.95044672778346\\
38	2.02304579803645\\
40	2.09220654883217\\
42	2.15823111686441\\
44	2.22138554612312\\
46	2.28190483140621\\
48	2.33999722783972\\
50	2.39584792078714\\
52	2.44962214560148\\
54	2.50146783842747\\
56	2.55151788971894\\
58	2.5998920625055\\
60	2.6466986283872\\
62	2.6920357660813\\
64	2.73599276020328\\
66	2.77865103182756\\
68	2.82008502717159\\
70	2.86036298637518\\
72	2.89954761070088\\
74	2.93769664344754\\
76	2.97486337735637\\
78	3.01109709920826\\
80	3.04644348058659\\
};
\addlegendentry{LB Tan. $\ell=2$};

\addplot [color=black,solid]
  table[row sep=crcr]{%
10	0\\
12	0.0487152850330279\\
14	0.0966202231886317\\
16	0.144093333110286\\
18	0.190893425616554\\
20	0.237064356125142\\
22	0.282601922095001\\
24	0.327511697115924\\
26	0.371797749092483\\
28	0.415464687164814\\
30	0.458517333331962\\
32	0.500960773718987\\
34	0.542800339445677\\
36	0.584041596503432\\
38	0.624690333433229\\
40	0.664752548946986\\
42	0.704234439310838\\
44	0.74314238559549\\
46	0.781482940842508\\
48	0.819262817198393\\
50	0.856488873065025\\
52	0.893168100313105\\
54	0.929307611602738\\
56	0.964914627852885\\
58	0.999996465898583\\
60	1.03456052637223\\
62	1.06861428184209\\
64	1.10216526523852\\
66	1.13522105859504\\
68	1.16778928212894\\
70	1.19987758368242\\
72	1.23149362854311\\
74	1.26264508965927\\
76	1.29333963826278\\
78	1.32358493490975\\
80	1.35338862094677\\
};
\addlegendentry{LB Tan. $\ell=4$};

\addplot [color=blue,dashdotted]
table[row sep=crcr]{%
10	6.71424551766612\\
12	6.96578428466209\\
14	7.17990909001493\\
16	7.36632221424582\\
18	7.53138146051631\\
20	7.67948009950545\\
22	7.81378119121704\\
24	7.93663793900257\\
26	8.04984854945056\\
28	8.1548181090521\\
30	8.25266543245025\\
32	8.34429590791582\\
34	8.43045255166553\\
36	8.51175265376738\\
38	8.58871463558226\\
40	8.66177809777199\\
42	8.73131903102506\\
44	8.79766152585376\\
46	8.86108690599539\\
48	8.92184093707449\\
50	8.98013957763916\\
52	9.03617361255349\\
54	9.09011241966429\\
56	9.14210705730255\\
58	9.19229281447077\\
60	9.24079133216196\\
62	9.28771237954945\\
64	9.33315535031062\\
66	9.37721053038855\\
68	9.41996017784789\\
70	9.46147944728616\\
72	9.5018371849023\\
74	9.54109661534952\\
76	9.57931593758001\\
78	9.61654884377899\\
80	9.65284497300198\\
};
\addlegendentry{Pal., $\ell=2$};

\addplot [color=brown,dotted, line width=1]
table[row sep=crcr]{%
	11	4.4594316186373\\
	12	4.58496250072116\\
	13	4.70043971814109\\
	14	4.8073549220576\\
	15	4.90689059560852\\
	16	5\\
	17	5.08746284125034\\
	18	5.16992500144231\\
	19	5.24792751344359\\
	20	5.32192809488736\\
	21	5.39231742277876\\
	22	5.4594316186373\\
	23	5.52356195605701\\
	24	5.58496250072116\\
	25	5.64385618977472\\
	26	5.70043971814109\\
	27	5.75488750216347\\
	28	5.8073549220576\\
	29	5.85798099512757\\
	30	5.90689059560852\\
	31	5.95419631038687\\
	32	6\\
	33	6.04439411935845\\
	34	6.08746284125034\\
	35	6.12928301694497\\
	36	6.16992500144231\\
	37	6.20945336562895\\
	38	6.24792751344359\\
	39	6.28540221886225\\
	40	6.32192809488736\\
	41	6.35755200461808\\
	42	6.39231742277876\\
	43	6.4262647547021\\
	44	6.4594316186373\\
	45	6.49185309632967\\
	46	6.52356195605701\\
	47	6.55458885167764\\
	48	6.58496250072116\\
	49	6.61470984411521\\
	50	6.64385618977472\\
	51	6.6724253419715\\
	52	6.70043971814109\\
	53	6.7279204545632\\
	54	6.75488750216347\\
	55	6.78135971352466\\
	56	6.8073549220576\\
	57	6.83289001416474\\
	58	6.85798099512757\\
	59	6.88264304936184\\
	60	6.90689059560852\\
	61	6.93073733756289\\
	62	6.95419631038687\\
	63	6.97727992349992\\
	64	7\\
	65	7.02236781302845\\
	66	7.04439411935845\\
	67	7.06608919045777\\
	68	7.08746284125034\\
	69	7.10852445677817\\
	70	7.12928301694497\\
};
\addlegendentry{Burst $\ell=1$};

\addplot [color=blue,dotted, line width=1]
table[row sep=crcr]{%
	11	6.24996664253042\\
	12	6.42692052890731\\
	13	6.58813643252831\\
	14	6.73614398642199\\
	15	6.8729114520047\\
	16	7\\
	17	7.11866845651243\\
	18	7.22994643762538\\
	19	7.33468666167423\\
	20	7.43360316476818\\
	21	7.52729974117494\\
	22	7.61629146034608\\
	23	7.70102118723671\\
	24	7.78187243514239\\
	25	7.85917948551151\\
	26	7.93323544291181\\
	27	8.00429871033851\\
	28	8.07259824215294\\
	29	8.13833784115552\\
	30	8.20169970093913\\
	31	8.26284734699907\\
	32	8.32192809488736\\
	33	8.37907511743363\\
	34	8.43440919324551\\
	35	8.48804019361763\\
	36	8.54006835338831\\
	37	8.59058536230344\\
	38	8.63967530643356\\
	39	8.68741548367319\\
	40	8.73387711297951\\
	41	8.77912595352186\\
	42	8.82322284711565\\
	43	8.86622419505639\\
	44	8.90818237863646\\
	45	8.94914613113165\\
	46	8.98916086781764\\
	47	9.02826897956599\\
	48	9.0665100947323\\
	49	9.10392131335333\\
	50	9.1405374170894\\
	51	9.17639105786095\\
	52	9.21151292771912\\
	53	9.24593191214371\\
	54	9.27967522866916\\
	55	9.31276855248959\\
	56	9.34523613048172\\
	57	9.37710088490221\\
	58	9.4083845078601\\
	59	9.43910754753061\\
	60	9.46928948696059\\
	61	9.49894881621566\\
	62	9.52810309853198\\
	63	9.55676903105995\\
	64	9.58496250072116\\
	65	9.61269863564227\\
	66	9.63999185257897\\
	67	9.66685590069898\\
	68	9.69330390205395\\
	69	9.71934838903593\\
	70	9.74500133908367\\
};
\addlegendentry{Burst $\ell=2$};

\addplot [color=black,dotted, line width=1]
table[row sep=crcr]{%
	11	10.8310366903167\\
	12	11.1108365852796\\
	13	11.3635298613028\\
	14	11.5937221151508\\
	15	11.804953164797\\
	16	12\\
	17	12.1810796870366\\
	18	12.3499893099915\\
	19	12.5082049581355\\
	20	12.6569533045298\\
	21	12.7972643779673\\
	22	12.9300111437637\\
	23	13.0559396495961\\
	24	13.1756923039849\\
	25	13.2898260769851\\
	26	13.3988268924532\\
	27	13.5031211266886\\
	28	13.6030848823436\\
	29	13.6990515332114\\
	30	13.7913179116004\\
	31	13.8801494202235\\
	32	13.9657842846621\\
	33	14.048437113584\\
	34	14.1283018972359\\
	35	14.205554546963\\
	36	14.2803550572803\\
	37	14.3528493556524\\
	38	14.4231708924135\\
	39	14.4914420132951\\
	40	14.5577751491638\\
	41	14.6222738513294\\
	42	14.6850336957894\\
	43	14.746143075765\\
	44	14.8056838986348\\
	45	14.8637322007356\\
	46	14.9203586913389\\
	47	14.9756292353427\\
	48	15.0296052827546\\
	49	15.0823442518296\\
	50	15.1338998717187\\
	51	15.1843224896399\\
	52	15.2336593468752\\
	53	15.2819548273047\\
	54	15.3292506816805\\
	55	15.3755862304194\\
	56	15.4209985473299\\
	57	15.4655226263771\\
	58	15.5091915333252\\
	59	15.5520365438681\\
	60	15.5940872696647\\
	61	15.6353717735212\\
	62	15.6759166748222\\
	63	15.71574724618\\
	64	15.7548875021635\\
	65	15.7933602808699\\
	66	15.83118731902\\
	67	15.8683893211814\\
	68	15.9049860236612\\
	69	15.9409962535514\\
	70	15.9764379833611\\
};
\addlegendentry{Burst $\ell=4$};

\addplot [color=brown,dashed]
table[row sep=crcr]{%
10	2.5737352452979\\
12	2.81017544111998\\
14	3.00564656314114\\
16	3.17862460762746\\
18	3.3300077203836\\
20	3.4660923806282\\
22	3.59064221440495\\
24	3.70523792115503\\
26	3.81146668969426\\
28	3.910447364017\\
30	4.00310552010772\\
32	4.09019768587916\\
34	4.17235171581219\\
36	4.25009534555321\\
38	4.32387636480587\\
40	4.39407788023758\\
42	4.46103017269694\\
44	4.52501996020828\\
46	4.58629771139705\\
48	4.64508349315495\\
50	4.70157168763026\\
52	4.75593483073454\\
54	4.808326761145\\
56	4.85888522292833\\
58	4.90773403181544\\
60	4.95498489054572\\
62	5.00073892115116\\
64	5.04508797217014\\
66	5.08811575548354\\
68	5.12989886028276\\
70	5.17050768139207\\
72	5.2100073365845\\
74	5.2484588546224\\
76	5.28592133112564\\
78	5.32245613539685\\
80	5.35813443056772\\
};
\addlegendentry{VT Tan. $\ell=1$};

\addplot [color=blue,dashed]
table[row sep=crcr]{%
10	1.3417885172482\\
12	1.7759983258019\\
14	2.1623720668286\\
16	2.49022499567306\\
18	2.78325414180469\\
20	3.04475876865798\\
22	3.26837601971807\\
24	3.45672268502473\\
26	3.61956535829396\\
28	3.76240168601457\\
30	3.88748548525368\\
32	3.99794903907477\\
34	4.09723281978626\\
36	4.18777250014644\\
38	4.27098112354258\\
40	4.34803237340095\\
42	4.42007042836499\\
44	4.48796843614551\\
46	4.55229159862112\\
48	4.6134682405098\\
50	4.67187829125494\\
52	4.72783310368139\\
54	4.78157105599477\\
56	4.8332867176413\\
58	4.88314930853129\\
60	4.9313042592693\\
62	4.97787532524981\\
64	5.02297037476445\\
66	5.06668543874225\\
68	5.10910633330113\\
70	5.15030994975994\\
72	5.1903657971295\\
74	5.22933724323843\\
76	5.26728228908209\\
78	5.30425415841236\\
80	5.34030186585498\\
};
\addlegendentry{VT Tan. $\ell=2$};

\addplot [color=black,dashed]
table[row sep=crcr]{%
10	0.0931094043914822\\
12	0.192645077942396\\
14	0.296096426555337\\
16	0.398229211592289\\
18	0.500029576247815\\
20	0.601890966290327\\
22	0.703540800159288\\
24	0.804856115615841\\
26	0.905841753798587\\
28	1.00656487310381\\
30	1.10706277257786\\
32	1.20729411500409\\
34	1.30719126218808\\
36	1.40672401333391\\
38	1.50590146839503\\
40	1.60473186482438\\
42	1.70319576521219\\
44	1.80125605407299\\
46	1.89888249217677\\
48	1.99606055026057\\
50	2.09278027438864\\
52	2.18902285778417\\
54	2.28475943672559\\
56	2.37995967602192\\
58	2.47459848127733\\
60	2.56865465829608\\
62	2.66210540780062\\
64	2.75492352633832\\
66	2.84707926240645\\
68	2.93854344530209\\
70	3.0292883205328\\
72	3.11928603393025\\
74	3.20850713527041\\
76	3.29692068330883\\
78	3.38449538533233\\
80	3.47120035012641\\
};
\addlegendentry{VT Tan. $\ell=4$};

\end{axis}
\end{tikzpicture}%